\documentclass[a4paper,10pt]{article}

\usepackage[utf8]{inputenc}
\usepackage[english]{babel}
\usepackage[T1]{fontenc}

\usepackage{amsmath}
\usepackage{amsfonts}
\usepackage{amssymb}
\usepackage{amsthm}
\usepackage{bbm}

\usepackage{mathtools}

\usepackage{setspace}
\usepackage{afterpage}

\usepackage{hyperref}
\newcommand{\Z}{{\mathbb{Z}}}

\newcommand{\R}{{\mathbb{R}}}
\newcommand{\Q}{{\mathbb{Q}}}
\newcommand{\N}{{\mathbb{N}}}

\DeclareMathOperator*{\argmin}{arg\,min}
\DeclareMathOperator*{\argmax}{arg\,max}
\DeclareMathOperator{\acc}{Acc}
\DeclareMathOperator{\vol}{Vol}

\renewcommand{\epsilon}{\varepsilon}
\renewcommand{\phi}{\varphi}

\theoremstyle{plain}
\newtheorem{thm}{Theorem}[section]
\newtheorem{lem}[thm]{Lemma}
\newtheorem{prop}[thm]{Proposition}

\theoremstyle{definition}
\newtheorem{defi}{Definition}

\newtheorem{remark}{Remark}

\usepackage{authblk}

\title{Structure function and fractal dissipation for an intermittent inviscid dyadic model}
\author{Luigi Amedeo Bianchi}
	\affil{Institut f\"ur Mathematik\\ Technische Universit\"at Berlin\\ \href{mailto:bianchi@math.tu-berlin.de}{bianchi@math.tu-berlin.de}}
\author{Francesco Morandin}
\affil{Dipartimento di Scienze Matematiche, Fisiche e Informatiche\\ Universit\`{a} degli Studi di Parma\\ \href{mailto:francesco.morandin@unipr.it}{francesco.morandin@unipr.it}}

\date{}

\begin{document}

\maketitle
\begin{abstract}
  We study a generalization of the original tree-indexed dyadic model
  by Katz and Pavlovi\'c for the turbulent energy cascade of
  three-dimensional Euler equation. We allow the coefficients to vary
  with some restrictions, thus giving the model a realistic spatial
  intermittency. By introducing a forcing term on the first component,
  the fixed point of the dynamics is well defined and some explicit
  computations allow to prove the rich multifractal structure of the
  solution. In particular the exponent of the structure
  function is concave in accordance with other theoretical and
  experimental models. Moreover anomalous energy dissipation happens
  in a fractal set of dimension strictly less than 3.
\end{abstract}

\section{Introduction}

Three dimensional turbulent fluids are far from being fully
understood, from a mathematical point of view. Even if we know the
equations governing the behaviour of the fluid, extracting the laws of
turbulence is extremely difficult. It is not surprising, then, that
many simplified models have been developed in the past years to
capture at least some aspects of turbulent fluids.  Among those the
shell models are of particular interest. Introduced by Novikov, they
have many variants. We recall here the dyadic model~\cite{DesNov1974}
and the GOY, as introduced by Gledzer~\cite{gledzer1973system} and
Ohkitani and Yamada~\cite{ohkitani1989temporal}.

The study of shell models in turbulence is well established in the
physics literature, in particular for the relative ease of numerical
simulation. A nice review of this is Biferale~\cite{Bif}.

The model we are interested in belongs to the family of dyadic shell
models, and was introduced by Katz and
Pavlovi\'c~\cite{MR2095627}. Its main feature is the tree structure of
the components, which allows to write a simplified wavelet description
of the Euler equations. (Conversely, the more common integer-indexed
shell models are constructed to be reminiscent of Littlewood-Paley
decomposition.)

Even if the motivations for these models are quite different, it is
also natural to see the tree models as generalizations of the usual shell
models. This has been done for example by Benzi and Biferale for the
GOY model in~\cite{benzi19971+} and in~\cite{barbiaflamor} for many
results that were proved about the dyadic in~\cite{BarFlaMor11}.

\paragraph{Anomalous dissipation.}

One of the main features of most inviscid shell models is the blow-up
of regularity, linked to the ``anomalous'' dissipation of energy. With
the latter we intend that the non-linear, formally conservative term,
``fires'' lumps of energy to smaller and smaller scales, making them
actually disappear. In passing from Euler to Navier-Stokes, the
introduction of a term corresponding to the viscosity of the fluid may
sometimes be enough to brake this phenomenon (as was proved for the
dyadic model with viscosity in~\cite{BarMorRom}), but the non-linear
term can be tailored to overcome thermal dissipation, in fact Tao
in~\cite{tao2015finite} used a shell model to prove that some averaged
versions of three-dimensional Navier-Stokes equation have blow-up.

Anomalous dissipation is connected to Onsager's conjecture on the
regularity of the solutions of Euler equation, discussed later on.

\paragraph{RCM, tree dyadic with repeated coefficients.}

In this paper, we build on the previous work in~\cite{barbiaflamor},
and consider a more general model, that still exhibits anomalous
dissipation of energy. The model will be introduced in the following
section. The main difference from the literature is that we allow the
coefficients of the non-linear term to depend on the nodes of the tree
not only through their generation. Every node $j$ of the tree has
$N=2^d=8$ children $j_1, j_2,\dots,j_N$ and interacts with each one of
them in the same way but for a coefficient
$c_{j_i}=2^{\alpha|j|}\delta_i$, where $\{\delta_1,\dots,\delta_N\}$
are fixed positive numbers that are repeated for all nodes $j$ and
$|j|$ is the generation of $j$. We call this the model with repeated
coefficients or RCM.

In the previous models from the literature the choice was
$\delta_i\equiv1$, and in many cases the solutions were uniform in
phase space and quite regular in physical space. Allowing for
different $\delta_i$'s forces spatial intermittency on the solutions,
thus yielding interesting results in terms of structure function and
singularities spectrum. From a physical point of view, we see this
generalization as a picture of the ``istantaneous'' Euler dynamics, as
explained in Remark~\ref{r:variable_coeff_instant_euler}.

\paragraph{Structure function.}

Structure functions are among the main objects studied in physics to
give a statistical description of the energy cascade in
turbulence. These are denoted by $S_p(r)$ and defined as the
$p$-moments of the velocity increments on the scale $r$. In his
cornerstone work on the theory of turbulence K41~\cite{K41},
Kolmogorov postulated that $S_p(r)\sim r^{p/3}$, but subsequent
numerical and experimental studies (for
example~\cite{anselmet1984high,belin1996exponents,benzi1996generalized,PhysRevE.59.5457})
did not fully agree with such prediction, showing instead a scaling
exponent nonlinearly dependent on $p$: $S_p(r)\sim r^{\zeta_p}$. This
discrepancy is evidence of a multifractal spectrum of the
singularities and is usually attributed to some spatial intermittency
of the energy cascade
(see~\cite{benzi1984multifractal,parisi1985singularity,meneveau1991multifractal,riedi1999multifractal}
and many references therein).

To cope with this discrepancy, in the past years there have been
several attempts to develop phenomenological models for the energy
cascade which are intermittent and self-similar. To cite just a few,
there is the log-normal model by Kolmogorov and
Obukhov~\cite{K62,obukhov1962some}, the random curdling by
Mandelbrot~\cite{mandelbrot1974intermittent}, the $\beta$-model by
Frisch, Sulem and Nelkin~\cite{frisch1978simple}, the random
$\beta$-model by Benzi, Paladin, Parisi and Vulpiani~\cite{benzi1984multifractal}, to the
more recent $\alpha$- and $p$-models
(see~\cite{meneveau1991multifractal} for an excellent review) and
finally, the log-Poisson model by She and L\'ev\^eque~\cite{she1994universal}.

Many of these models actually exhibit a concave $\zeta_p$, thus
yielding rich multifractal structure, but none are obtained as
solutions of a dynamical model. On the other hand, Benzi, Biferale and Parisi~in~\cite{benbifpar93} deduce a plausible $\zeta_p$ for
the stationary distribution of the GOY shell model, but their
derivation is not rigorous.

One of the main results of this paper is that the constant solution
of RCM is a self-similar, multifractal function that truly exhibits a
non-linear scaling exponent of singularities, with a concave graph not
dissimilar from those coming from numerical experiments.  (On the
contrary, the dyadic shell model and the tree dyadic model both agree
with Kolmogorov theory and have $\zeta_p=p/3$.)

As already stated, this is linked to spatial intermittency of the
energy cascade. Truly, one can introduce the measure associated with
turbulent energy cascade of the solution, and prove that this measure
is itself a self-similar, multifractal, multiplicative cascade.

\paragraph{Onsager conjecture.}

In~\cite{onsager1949statistical}, while trying to
understand the phenomenon of energy dissipation in three-dimensional
turbulent fluids for vanishing viscosity, Lars Onsager stated the
conjecture that bears his name: that solutions of the incompressible
Euler equations are energy preserving if they have a H\"older
regularity greater than $1/3$ and that for every H\"older exponent
$\alpha\leq1/3$ there exists a weak solution of the Euler equation in
$C^\alpha$ that dissipates energy.

The first half of this conjecture has been proven by Constantin, E and
Titi~\cite{constantin1994onsager} for three-dimensional Euler
equations, in the setting of Besov spaces, building on a previous work
by Eyink~\cite{eyink1994energy}. The second half lead to the
development of many partial results, in particular by Buckmaster, De
Lellis, Isett and
Székelyhidi~\cite{buckmaster2015dissipative,buckmaster2015onsager,Isett2017},
but it is still open.

It is worth noting that the unique constant solution of our model
always exhibits anomalous dissipation and it has H\"older regularity
$h\leq1/3$. In particular $h<1/3$ if the coefficients $\delta_i$ are
not all equal. (See Theorem~\ref{t:Cs_and_zetap_RCM}.) This is in
accordance with Onsager's conjecture. Moreover, if we introduce the
local H\"older exponent $s(x)$ for each point $x$, then $s(x)\geq h$
and it is possible to compute its multifractal spectrum and to show
that anomalous dissipation occurs at all points $x$ for which
$s(x)\leq 1/3$.

\paragraph{Constant solutions.}

One serious drawback of RCM is that it is mathematically hard to deal
with. In fact we cannot prove significant results for the general
solution of the problem. Instead we introduce a constant forcing term
on the first component and look for \emph{constant} solutions.

The fact that finite-energy, constant solutions exist, is \emph{per
se} an interesting proof of anomalous dissipation, but ---what is more
important--- the constant solution can be made completely explicit, and
its structure analysed in every detail.

One might wonder if considering only constant solutions is too
restrictive, but we stress that they are an interesting first step
that motivates further study of models on trees with variable
coefficients. Moreover it is reasonable to conjecture that the
constant solution is an attractor (as is the case for the dyadic shell
model, see Cheskidov et al.~\cite{CheFriPav2010}), making its
properties even more interesting. The next natural step would be to
study solutions that are not constant in time but statistically
stationary, in some sense. We believe that many properties of constant
solutions are universal and hence would hold also for stationary
solutions.

\paragraph{Main results.}

For the sake of clarity, we state here the main results of the paper
in the physically meaningful case, that is $d=3$ and $\alpha=5/2$.
The complete statements and the proofs can be found in
Section~\ref{s:general_existence},~\ref{s:structure_function}
and~\ref{s:fractality}.

\begin{thm}
There exists a unique constant finite-energy solution for the RCM. The
exponents of the structure function corresponding to this solution are
given by
\[
\zeta_p = \frac{p}{3} +\frac{p}{2}(\ell_{3/2}-\ell_{p/2})
,\qquad p\geq0
\]
where $\ell_{s}$ is a function of $s$ that depends on the repeated
coefficients $\delta_i$'s: it is constant if they are all equal, while
otherwise it is strictly increasing with finite limits at
$\pm\infty$. In the latter case, the function $\zeta_p$ is strictly
concave and has an oblique asymptote. Moreover if the ratio between
the maximum of $\delta_i^{3/2}$, $i=1,2,\dots,N$ and their average is
less then 2, then $\zeta_p$ is increasing for all $p$.
\end{thm}

This depicts a model with spatial intermittency, as the scaling
exponents, for $p>3$, lie below the Kolmogorov's $p/3$ line.

To study the geometry of anomalous dissipation we associate each index
$j\in J$ to one cube $Q_j$ of side $2^{-|j|}$ in the dyadic lattice
$\bigcup_n(2^{-n}[0,1])^3$ and identify a non-negative term $F_j$
measuring the energy dissipated inside the cube $Q_j$, with the
property that
\[
\sum_{|j|=n}F_j=1.
\]
\begin{thm}
Suppose the repeated coefficients $\delta_i$'s are not all
equal. Then there is a set $\mathcal H\subset[0,1]^3$ of Hausdorff
dimension strictly less than 3 such that
\[
\sum_{\substack{|j|=n\\Q_j\cap\mathcal H\neq\emptyset}}F_j=1
,\qquad
\sum_{\substack{|j|=n\\Q_j\cap\mathcal H=\emptyset}}F_j=0
,\qquad n\geq1
\]
and
\[
\lim_{n\to\infty}\vol\biggl(\bigcup_{\substack{|j|=n\\Q_j\cap\mathcal H\neq\emptyset}}Q_j\biggr)=0.
\]
\end{thm}

The structure of the paper is the following: in
Section~\ref{s:physical_model} we introduce our model and discuss its
physical meaning, with some additional details presented in
Appendix~\ref{app:physical}. In Section~\ref{s:general_existence} we
prove existence and uniqueness of the constant solution, then we move
on to determine the form of the exponent in the structure function and
discuss its properties in Section~\ref{s:structure_function}. Finally,
in Section~\ref{s:fractality}, we prove the multifractality results
for the anomalous dissipation of energy.

\section{The models}\label{s:physical_model}
This section is devoted to the presentation of the dyadic tree model
introduced by Katz and Pavlovi\'c in~\cite{MR2095627} and studied
again in Barbato et al.~\cite{barbiaflamor} and to its generalization,
which is the main model of this paper.

These models specify the dynamics in terms of some coefficients
$(v_j(t))_j$, indexed by a tree $J$. The equations have some likeness
to those one would get with any wavelet decomposition of Euler
equations:
\[
\frac d{dt}v_j(t)
=\sum_{k,l\in J}C_{j,k,l}v_k(t)v_l(t).
\]

In the previous works the model has been studied as an abstract
formulation, but in the present work we would like to investigate also
some geometric properties of the physical ``solution'', in the physical
space.

To this end, we prove rigorous statements for the abstract model, but
give also non-rigorous consequences for a physical ``solution'' which
we imagine to be recomposed from the coefficients $v_j(t)$'s through
any orthonormal family $(\psi_j(x))_j$ of wavelets on a cube
$Q_\emptyset$ of $\R^d$.
\begin{equation}\label{e:wavelet_recomposition_t}
v(t,x)
\coloneqq \sum_{j\in J}v_j(t)\psi_j(x)
,\qquad t\geq0,x\in Q_\emptyset.
\end{equation}

We do not explicitly choose the wavelets, but try to deduce universal
consequences, which would not depend on the choice. In particular,
for our purposes, the physical ``solution'' $v(t)$ will be a scalar%
\footnote{It may seem confusing that $v$ is scalar, but the results
for a vectorial field would essentially be the same. In fact the
dynamics is not deduced rigorously from Euler equations. Instead the
abstract model is introduced at the level of the coefficients $v_j(t)$
in such a way to ensure the cascade of energy. The reconstructed field
$v(t,x)$ is then studied only from the point of view of its
regularity. If we chose vectorial wavelets instead, all the results could be
easily restated for the vectorial case, with a more cumbersome
notation but without any significant change in the results.}
field whose regularity is what we propose to study.

Consider a solution which is stationary in some sense, $v(t,x)\approx
u(x)$ for all~$t$. The structure function of order $p$ of $u$ is
defined by
\[
S_p(r)\coloneqq \int_{Q_\emptyset}\left\langle|u(x)-u(y)|^p\right\rangle_ydx
\]
where $\langle\cdot\rangle_y$ denotes the average on the points $y$
such that $|y-x|=r$.

This is a very popular tool to study turbulence in fluid mechanics.
In particular one often considers the infinitesimal behaviour of
$S_p(r)$ as $r\to0$, introducing the exponents of the structure
function, that is
\begin{equation}\label{e:zeta_p_def}
\zeta_p\coloneqq -\lim_{n\to\infty}\frac1n\log_2S_p(2^{-n}).
\end{equation}

It is known that $\zeta_p$ is linked to the Besov norms
$B_p^{s,\infty}$, which in turn can be computed from the wavelet
coefficients of $u$ in an universal fashion, not depending on the
actual wavelet basis chosen.

In this section, after an introduction of the abstract model, we will
link it to a physical solution, define and compute some Besov norms of
the latter and finally deduce the formula of $\zeta_p$ in terms of the
solution to the abstract model.

\subsection{Abstract model}

Let $d$ be the space dimension and let $N=2^d$. Consider the following
set with its natural tree structure:
\[
	J\coloneqq \bigcup_{n=0}^\infty\{1,2,\dots,N\}^n
	=\{\emptyset,1,2,\dots,N,(1,1),(1,2),\dots\}.
\]
For all $j=(j^1,j^2,\dots,j^m),k=(k^1,k^2,\dots,k^n)\in J$, we define
the append operator $jk\coloneqq (j^1\dots,j^m,k^1,\dots,k^n)\in J$, the size
operator $|j|\coloneqq m\in\N$, the partial ordering $j\leq k$ if and only if $k=jh$ for
some $h\in J$ (with $j<k$ if moreover $|h|>0$), the father operator
$\bar\jmath\in J$ such that $\bar\jmath<j$ and $|\bar\jmath|=|j|-1$
and the offspring set of $j$, $\mathcal O_j\coloneqq \{k\in J:\bar k=j\}$.

Our model is given by the following equations
\begin{equation}\label{e:model}
v_j'(t)=c_jv_{\bar\jmath}^2(t)-\sum_{k\in\mathcal O_j}c_kv_j(t)v_k(t)
,\qquad j\in J,\quad t\geq0
\end{equation}
where $c_j=d_j2^{\alpha |j|}$, $\alpha>0$, $d_j>0$ for $j\in J$,
$d_\emptyset=1$ and $v_{\bar\emptyset}(t)\equiv f$.

It generalizes the model introduced by Katz and Pavlovi\'c
in~\cite{MR2095627}, where $f=0$ and $d_j=1$ for all $j\in J$.

The parameter $\alpha$ is left free in all statements, but from a
physical point of view, some heuristic arguments based on Euler
dynamics suggest to fix $\alpha=\frac d2+1$, which is what the other
authors also used. See for example works by Katz, Pavlovi\'c,
Friedlander and
Cheskidov~\cite{MR2095627,FriPav,CheFriPav2007,MR2522972}. Recently it
was proved rigorously in Barbato et al.~\cite{barbato2015global} that
$\alpha\leq\frac52$ for a Littlewood-Paley decomposition of the true
three-dimensional Euler dynamics.

The generalization to variable $d_j$'s is very important. As we will
see it completely changes the behaviour of anomalous dissipation and
makes the function $\zeta_p$ strictly concave (as it should be,
according to the most important numerical simulations of realistic
turbulence models).

\begin{remark}\label{r:variable_coeff_instant_euler}
We believe this generalization to be well justified from a physical
point of view. When passing from a detailed description of Euler
equations to any shell model of turbulence, many components (either
Fourier or wavelets) are merged inside any single component of the
shell model, thus the nonlinear interaction between adjacent shell
components cannot be known precisely, and actually it depends on how
the energy of the shell is distributed among the original components.
In~\cite{barbato2015global} for example a shell model is \emph{rigorously}
deduced from Euler equations and truly the coefficients
$\phi_{l,m,n}(t)$ of the nonlinear interaction turn out to be
complicated, to depend on time and on the solution itself and they
only allow to be studied by the bound
$|\phi_{l,m,n}(t)|\leq2^{(5/2)\min(l,m,n)}$. This means that at any
fixed time $t$ the true Euler dynamics, seen through a realistic
shell model, have ``instantaneous'' coefficients of interactions which
are all different and only statistically behave like $2^{(5/2)n}$.

In this sense looking for constant solutions of the models with
variable coefficients identifies a very large class of fields among
which we expect to find the solutions of the true Euler equations
which in some sense are stationary or stable with respect to time
evolution.
\end{remark}

It would be really important to have a complete generality of the
variable coefficients. In our model we always consider $|\log d_j|$
bounded and the more general results are proved in this setting.
Nevertheless explicit computation of many quantities is possible only
in the special case that the same fixed $N$ coefficients
$\delta_\omega$ appear in every set of the form $\{d_k:k\in\mathcal
O_j\}$. We call this the model with repeated coefficients or RCM (see
Definition~\ref{def:rep_coeff}) and our most interesting and
meaningful results are restricted to this model.

\subsection{Physical space}\label{s:physical_space}

Three-dimensional Navier-Stokes equations have been studied several times
by means of multiresolution analysis or wavelet decomposition
(see~\cite{deriaz2006divergence,stevenson2011divergence} and
references therein). The typical expression for the velocity field is
\[
v(x,t)
=\sum_{Q\in\mathcal Q}\sum_{a\in\mathcal A}v_Q^a(t)\psi_Q^a(x),
\]
where $\mathcal Q$ is the set of the dyadic cubes inside
$Q_\emptyset\coloneqq[0,1]^3$, $\psi_Q^a$ is a rescaling essentially supported
on $Q$ of the ``mother'' $\psi_{Q_\emptyset}^a$ of the wavelets, and
$\mathcal A$ is a fixed, finite set of indices that allow these
wavelets to be a basis of some suitable functional space on
$Q_\emptyset$. For example, to get a basis of $L^2(\R^3,\R^3)$, one
must provide 21 different ``mother'' wavelets (7 for each component)
and the same number is required for divergence-free vector fields.

In the dyadic models of turbulence the phase-space $\mathcal
Q\times\mathcal A$ is simplified to $\mathcal Q$ (in the case
of~\cite{MR2095627} or our dyadic model on a
tree~\cite{barbiaflamor,bianchi2013uniqueness}) or even a quotient of
$\mathcal Q$ (in classical shell models of turbulence that follow
Littlewood-Paley decomposition). The non-linear interaction is
constructed anew to be elementary but retain some of the main
properties of the bilinear term in Euler equations.

In the present work in particular we identify $\mathcal Q$ with the
tree $J$ through an isomorphism for which the relation $\subseteq$ on
$\mathcal Q$ corresponds to $\geq$ on $J$.

More precisely, let $Q_\emptyset$ be the unit cube of $\R^d$, which is
divided into $N=2^d$ cubes of side $\frac12$ which are labelled
$Q_1,Q_2,\dots,Q_N$ in some fixed way.

To each $j\in J$ we associate one cube $Q_j$ of side $2^{-|j|}$.
Above we defined $Q_j$ for $|j|=0,1$. Then recursively, each cube
$Q_j$ is divided into $N$ cubes of half side labelled
$Q_{j1},\dots,Q_{jN}$ following the same ordering as for $j=\emptyset$
in such a way that for all $j,k\in J$ the homothety that maps
$Q_\emptyset$ to $Q_j$ also maps $Q_k$ to~$Q_{jk}$.

Then for a.e.~point $x\in Q_\emptyset$ it is well defined the sequence
$\emptyset=x_0<x_1<x_2<\dots$ of elements of $J$ such that $|x_n|=n$
and $x\in Q_{x_n}$, and we will identify $x$ with $(x_n)_{n\geq0}$
when convenient.

Consider a real function $\psi_\emptyset$ on $Q_\emptyset$, the
``mother'' of the wavelets and for all $j\in J$, let
$\psi_j(x)\coloneqq 2^{d|j|/2}\psi_\emptyset(2^{|j|}x+\theta_j)$, where
$\theta_j$ is such that $\psi_j$ is supported on $Q_j$, the rescaling
being the correct one to have all $\psi_j$'s with the same $L^2$-norm.

Given this family of wavelets, we can associate a real function
$v(t,x)$ on $Q_\emptyset$ to any solution $(v_j(t))_j$ of the abstract
model~\eqref{e:model} through
equation~\eqref{e:wavelet_recomposition_t}.

The regularity in space of the field $v(t,\cdot)$ can be studied by
introducing suitable norms on the set of functions from $J$ to
$\R$. In particular, given $u:J\to\R$, $p\in[1,\infty]$ and $s\in\R$,
define a sequence $(\epsilon_n)_{n\geq0}$ by
\[
\epsilon_n
\coloneqq\begin{cases}
2^{ns}2^{dn(\frac12-\frac1p)}\Bigl(\sum_{|j|=n}|u_j|^p\Bigr)^{1/p} & p<\infty\\
2^{ns}2^{dn/2}\max_{|j|=n}|u_j| & p=\infty.
\end{cases}
\]

Then (see Meyer~\cite{mey92})
\begin{equation}\label{e:besov_cond_epsilon}
\sum_{j\in J}u_j\psi_j\in B_p^{s,q}
,\qquad\textup{if and only if}\qquad
\epsilon\in l^q(\N).
\end{equation}

With this identification of the Besov spaces at hand, we formally
introduce the spaces corresponding to the usual function spaces $H^s$,
$W^{s,p}$ and $C^s$ for sequences of real numbers indexed by $J$.

\begin{defi}\label{d:function_spaces}
For all $s\in\R$ we introduce the space $H^s$ of the maps $u:J\to\R$
such that the norm
\[
\|u\|_{H^s}
\coloneqq \Bigl(\sum_{j\in J}2^{2s|j|}u_j^2\Bigr)^{1/2}
\]
is finite. In particular let $H\coloneqq H^0=l^2(J)$.

Moreover, for all $s\in\R$ and $p\geq1$ we introduce the space
$W^{s,p}$ of the maps $u:J\to\R$ such that the norm
\[
\|u\|_{W^{s,p}}
\coloneqq 
\Bigl(\sum_{j\in J}2^{ps|j|}2^{d(\frac p2-1)|j|}|u_j|^p\Bigr)^{1/p}
\]
is finite. In particular $W^{s,2}=H^s$.

Finally, for all $s\in(0,1)$ we introduce the space $C^s$ of the maps
$u:J\to\R$ such that
\[
\sup_{n\geq1}\biggl(ns+\frac12dn+\max_{|j|=n}\log_2|u_j|\biggr)
\]
is finite.
\end{defi}
By condition~\eqref{e:besov_cond_epsilon}, these spaces correspond to
the usual ones for the recomposed function $\sum_{j\in J}u_j\psi_j$.

To make explicit the link between Besov norms and the exponents of the
structure function, we refer to the works by Eyink~\cite{eyink1995besov} and by
Perrier and Basdevant~\cite{perrier1996besov}. In the latter it is proven that if $\zeta_p$
is defined as usual by~\eqref{e:zeta_p_def}, then
\[
\zeta_p
=\sup\{s<p:u\in B_p^{s/p,\infty}\}.
\]
Thus if $u(x)=\sum_{j\in J}u_j\psi_j(x)$, by
condition~\eqref{e:besov_cond_epsilon},
\begin{equation}\label{e:defzetap}
\zeta_p
=\min\biggl\{p;d-\frac p2d-\limsup_{n\to\infty}\frac1n\log_2\sum_{|j|=n}|u_j|^p\biggr\}.
\end{equation}

In Appendix~\ref{app:physical} we give some other argument, not fully
rigorous, to show that this is indeed the correct exponent.

\section{Well-posedness and regularity}\label{s:general_existence}
In this section we will deal with the main model~\eqref{e:model} in
the abstract setting of the dyadic model on a tree. After some general
results we will restrict ourselves to the repeated coefficients model
and get a deeper understanding in that case.

Recall that $H=l^2(J)$.

\begin{defi}A \emph{componentwise solution} is a family $v=(v_j)_{j\in J}$ of
non-negative differentiable functions such that~\eqref{e:model} is
satisfied. A \emph{Leray solution} is a componentwise solution in
$L^\infty(\R^+;H)$.
\end{defi}
It has been proved in~\cite{barbiaflamor} that if $d_j\equiv1$ then
for any initial condition with non-negative components, there exists
at least one Leray solution. The argument is classical by Galerkin
approximations. The generalization to the model of this paper is
straightforward. Uniqueness of solutions is an open problem even for
the model with $d_j\equiv1$ and is a subtle matter. Uniqueness in fact
does not hold if one drops the non-negativity condition, but it is not
easy to exploit that hypothesis. One way to do that is a trick
presented in~\cite{BarFlaMor2010CRAS}, but the required estimates of
terms of the kind $\int_0^tX_n^3(s)ds$ for large $n$ are difficult to
generalize to other settings. (The more promising attempts for the dyadic
can be found in~\cite{BarMor12} and~\cite{BarEtAl2016}.) In the case of the tree dyadic model with $d_j\equiv1$, weak uniqueness is proven for a stochastically perturbed version in~\cite{bianchi2013uniqueness}.

\subsection{Constant Leray solutions}

From now on we will consider only Leray solutions $u=(u_j)_{j\in J}$,
not depending on time, that is

\[
0=c_ju_{\bar\jmath}^2-\sum_{k\in\mathcal O_j}c_ku_ju_k
,\qquad j\in J,
\]
yielding the fundamental recursion
\begin{equation}\label{e:recur_uj}
d_ju_{\bar\jmath}^2
=2^\alpha\sum_{k\in\mathcal O_j}d_ku_ju_k
,\qquad j\in J,
\end{equation}
because of the choice of the coefficients $(c_j)_{j\in J}$ given for this model.

One could try to find such a solution using~\eqref{e:recur_uj}
recursively, but there are two difficulties. Firstly, with
$u_{\bar\jmath}$ and $u_j$ given, the $N=2^d$ variables $u_k$ for
$k\in\mathcal O_j$ are not fixed by this single equation: there are
$N-1$ degrees of freedom left in their choice. Secondly, it is
difficult to prove that any such solution really belongs to $H$.  In
fact under some technical hypothesis, it will turn out that there
exists a unique Leray solution, so all choices but one give sequences
of numbers $u_j$ satisfying the recursion but not belonging to~$H$.

Both difficulties can be overcome by a sort of pull-back technique,
using the recursion \emph{backwards}. We will arbitrarily fix $u_j$ for
all $j\in J$ with given large generation $|j|=n$, then compute $u_k$
for the lower generations $|k|<n$ and then let $n\to\infty$, finally
proving convergence by compactness.

We will need to introduce the new variables $q_j$'s. Given $u$
satisfying the recursion~\eqref{e:recur_uj}, let
\begin{equation}\label{e:qj_def}
	q_j\coloneqq\log_2\biggl(\frac{u_j}{u_{\bar\jmath}\sqrt{d_j}}\biggr)
	, \qquad j\in J.
\end{equation}
Then $u$ can be recovered from $q=(q_j)_{j\in J}$, by
\begin{equation}\label{e:u_from_q}
	u_j=f2^{\sum_{h\leq j}q_h}\prod_{k\leq j}\sqrt{d_k}.
\end{equation}
The recursion~\eqref{e:recur_uj} rewrites equivalently in terms of $q$ as
\begin{equation}\label{e:recur_qj}
	q_j
	=-\frac12\alpha-\frac12\log_2\biggl(\sum_{k\in\mathcal O_j}d_k^{3/2}2^{q_k}\biggr).
\end{equation}

Before stating the theorem of existence, let us detail the
construction of the asymptotic Leray solution.

Let us fix $x\in\mathbb R$, define $q^{(n)}=(q^{(n)}_j)_{j\in
  J}$ for $n\geq1$ by
\begin{equation}\label{e:recur_qnj1}
	\begin{split}
		q^{(n)}_j	&\coloneqq 0 ,\qquad |j|>n\\
		q^{(n)}_j	&\coloneqq x ,\qquad |j|=n,
	\end{split}
\end{equation}
and then, recursively as $|j|$ decreases,
\begin{equation}\label{e:recur_qnj}
	q^{(n)}_j
	\coloneqq-\frac12\alpha-\frac12\log_2\biggl(\sum_{k\in\mathcal O_j}d_k^{3/2}2^{q_k^{(n)}}\biggr)
	,\qquad |j|<n.
\end{equation}
Finally, if the limit exists, we define
\[
	\tilde q_j\coloneqq\lim_{n\to\infty}q^{(n)}_j
	,\qquad j\in J,
\]
and $\tilde u$ from $\tilde q$ by~\eqref{e:u_from_q},
\begin{equation*}
	\tilde{u}_j \coloneqq f2^{\sum_{h\leq j}\tilde{q}_h}\prod_{k\leq j}\sqrt{d_k},\qquad j\in J.
\end{equation*}

\begin{remark}
It should be noted here that a solution of the above form is
reminiscent of the self-similar functions obtained by multiplicative
cascades of wavelet coefficients, which have been introduced from a
physical point of view by Benzi et al.~\cite{benzi1993random} and
mathematically formalized by Arneodo, Bacry and Muzy
in~\cite{arneodo1998random}. (See also
Riedi~\cite{riedi1999multifractal} for a comprehensive treatment of
multifractal processes, Jaffard~\cite{jaffard1997multifractal} for
detailed derivation of self-similar function in dimension $d$ and
Barral, Jin and Mandelbrot~\cite{barjinman10} for recent developments 
and more references.) The main difference is that here the multipliers
$2^{\tilde{q}_j}\sqrt{d_j}$ for $j\in J$ are not i.i.d.~random
variables, but a family of positive numbers (bounded from above and
away from zero).
\end{remark}

We can now state a first simple existence result.
\begin{thm}\label{t:existence_hr}
Suppose that the positive coefficients $(d_j)_{j\in J}$ are globally boun\-ded from above
and away from zero, that is
\[
\sup_{j\in J}\log_2d_j-\inf_{j\in J}\log_2d_j=:L<\infty
\]
Then there exists a constant componentwise solution $\tilde u$
of~\eqref{e:model} such that its coefficients $\tilde q_j$ defined as
in equation~\eqref{e:qj_def} satisfy recursion~\eqref{e:recur_qj} and
are bounded.

Moreover $\tilde u\in H^r$ for all
\begin{equation}\label{e:bd_r_existence}
r<\frac13\biggl(\alpha-\frac d2\biggr)-L.
\end{equation}
In particular, if $\alpha>\frac d2$ and

\[
\frac{\sup_{j\in J}d_j}{\inf_{j\in J}d_j}\leq2^{\frac13(\alpha-d/2)},
\]
then there exists a constant Leray solution.
\end{thm}

\begin{remark}
	We would like to stress here that we do not claim that condition~\eqref{e:bd_r_existence} is
	sharp, nevertheless it defines a class suitable to prove uniqueness.
\end{remark}

\begin{proof}[Proof of Theorem~\ref{t:existence_hr}]
Let $t\coloneqq\sup_{j\in J}\log_2d_j$ and $s\coloneqq\inf_{j\in J}\log_2d_j$ with $t-s=L$.

For any $n\geq1$, define $q^{(n)}$ as in~\eqref{e:recur_qnj1}
and~\eqref{e:recur_qnj}, starting with some $x\in\mathbb R$ that will
be fixed in the sequel. Let $a<b$ be given real numbers. If
$q^{(n)}_k\in[a,b]$ for $k\in\mathcal O_j$, then %
by~\eqref{e:recur_qnj}
\[
-\frac12\alpha-\frac12d-\frac12b-\frac34t
\leq q_j^{(n)}
\leq-\frac12\alpha-\frac12d-\frac12a-\frac34s, 
\]
and by letting
\[
a=-\frac13(\alpha+d)-t+\frac12s
,\qquad\text{and}\qquad
b=-\frac13(\alpha+d)-s+\frac12t,
\]
we get $q^{(n)}_j\in[a,b]$. Thus if $x$ is chosen inside $[a,b]$, by
induction all the components lie inside the same interval.

By compactness of $[a,b]$ and a diagonal extraction argument, we can choose a
subsequence $(n_i)_i$ such that $q^{(n_i)}_j$ converges for all $j\in
J$ to some number $\tilde q_j\in[a,b]$. The family $\tilde q=(\tilde
q_j)_{j\in J}$ satisfies recursion~\eqref{e:recur_qj} by
construction. Then $\tilde u$ obtained from $\tilde q$ by~\eqref{e:u_from_q}
is a constant componentwise solution.

Finally, if $r$ satisfies condition~\eqref{e:bd_r_existence}, then
\begin{equation*}
\begin{split}
\|\tilde u\|_{H^r}^2
&=\sum_{j\in J}2^{2r|j|}\tilde u_j^2
=\sum_{j\in J}f^22^{2r|j|}\prod_{k\leq j}d_k2^{2\sum_{h\leq j}\tilde q_h}\\
&\leq f^2\sum_{j\in J}2^{(2r+t+2b)|j|}
=f^2\sum_{i=0}^\infty2^{(2r+t+2b+d)i}
<\infty, 
\end{split}
\end{equation*}
where the last inequality holds because
$2r+t+2b+d=2r+2L-\frac23\alpha+\frac13d<0$.
\end{proof}

Uniqueness of constant solutions holds in a very large class, namely,
the union of $H^r$ for all $r\in\mathbb R$.

\begin{thm}\label{t:atmostonehr}
	Under the same hypothesis of Theorem~\ref{t:existence_hr}, for
	all $s\in\mathbb R$ there exists at most one constant
	componentwise solution in $H^s$.
\end{thm}

\begin{proof}
	Let $u$ be the solution to~\eqref{e:model} given by
        Theorem~\ref{t:existence_hr} and let $u'$ be a componentwise
        solution different from $u$.  Let $q$ be defined from $u$ as
        in equation~\eqref{e:qj_def} and $p$ be analogously defined
        from $u'$. We will show that since the coefficients $q_j$'s
        are bounded, then $u'$ cannot belong to $H^r$ for any $r$.

        Take $j_0\in J$ such that $q_{j_0}\neq p_{j_0}$ is the minimal
        generation $|j_0|$ where $p$ and $q$ differ. Suppose that
        $p_{j_0}=q_{j_0}+\epsilon_0$, with $\epsilon_0>0$ (the other
        case being analogous).  We can define recursively the sequence
        $(j_n)_{n\geq0}$ in $J$ by
        \[
        j_{n+1}= \begin{cases} \displaystyle\argmin_{k\in\mathcal
        O_{j_n}}(p_k-q_k) & n \text{
        even} \\ \displaystyle\argmax_{k\in\mathcal O_{j_n}}(p_k-q_k)
        & n \text{ odd}, \end{cases} \] and
        let \[ \epsilon_n \coloneqq p_{j_n}-q_{j_n} ,\qquad
        n\geq1.  \]
	
	By~\eqref{e:recur_qj}, both $p$ and $q$ satisfy
	\[
	\sum_{k\in\mathcal O_j}d_k^{3/2}2^{q_k+2q_j}=2^{-\alpha}
	,\qquad j\in J,
	\]
	hence
	\[
	\min_{k\in\mathcal O_j}(p_k+2p_j-q_k-2q_j)
	\leq0
	\leq\max_{k\in\mathcal O_j}(p_k+2p_j-q_k-2q_j),
	\]
	yielding that
	\[
	\min_{k\in\mathcal O_j}(p_k-q_k)
	\leq-2(p_j-q_j)
	\leq\max_{k\in\mathcal O_j}(p_k-q_k).
	\]
	These inequalities hold for all $j_n$, so that
	\begin{equation*}
		\begin{cases}
			\epsilon_{n+1}&\leq-2\epsilon_n, \qquad n\textup{ even}\\
			\epsilon_{n+1}&\geq-2\epsilon_n, \qquad n\textup{ odd},
		\end{cases}
	\end{equation*}
	hence, for all $n$ even we have $\epsilon_{n+2}\geq 4\epsilon_n$ and
	 $\epsilon_n\geq2^n\epsilon_0$. Moreover
	\[
	\epsilon_{n+1}+\epsilon_{n+2}
	\geq \epsilon_{n+1}-2\epsilon_{n+1}
	=-\epsilon_{n+1}
	\geq 2\epsilon_n
	\geq 2^{n+1}\epsilon_0
	,\qquad n\textup{ even},
	\]
	yielding that for all $n$ even, $\sum_{i=0}^n\epsilon_i
	\geq2^{n-1}\epsilon_0$.
	
	Since the coefficients $q_j$ are bounded by
        Theorem~\ref{t:existence_hr}, then for $n$ even we have
        \[
	       \sum_{i=0}^np_{j_i}\geq2^{n-1}\epsilon_0-nc 
        \]  
        and
        hence by~\eqref{e:u_from_q} $u'_{j_n}\geq C2^{\lambda^n}$ for
        $n$ even and large, with suitable constants $\lambda>1$ and
        $C>0$, yielding that $u'$ cannot belong to $H^s$ for any $s$.
\end{proof}

\subsection{Model with repeated coefficients}

From here on we will restrict ourselves to the model with repeated
coefficients, which allows for direct computation of many quantities
while still showing interesting features like intermittency and a
multifractal structure function.

\begin{defi}\label{def:rep_coeff}
We say that the model has repeated coefficients and call it RCM if the
set $\{d_k:k\in\mathcal O_j\}$ (considered with multiplicities) does
not depend on $j$. In this case we pose
$\{\delta_\omega:\omega\in\Omega\}=\{d_k:k\in\mathcal O_j\}$ for all
$j\in J$, for some $\Omega$ of cardinality $N$. If moreover all the
$\delta_\omega$ are equal we say that the model is flat.
\end{defi}

We also introduce the log-$s$-norm of the
coefficients, that will be used often. For $s\in\R\setminus\{0\}$ let
\begin{equation}\label{e:def_ell}
	\ell_s
	\coloneqq\frac1s\log_2\biggl(\frac1N\sum_{\omega\in\Omega}\delta_\omega^s\biggr).
\end{equation}
This can be completed with
\begin{gather*}
	\ell_0
	\coloneqq\frac1N\sum_{\omega\in\Omega}\log_2\delta_\omega,\\
	\ell_{-\infty}
	\coloneqq\lim_{s\to-\infty}\ell_s =\min_\omega\delta_\omega
	\qquad\textup{and}\qquad
	\ell_\infty \coloneqq\lim_{s\to\infty}\ell_s =\max_\omega\delta_\omega,
\end{gather*}
to get a bounded, non-decreasing and continuous function $\ell$ on
$[-\infty,\infty]$. Moreover, $\ell$ is constant if and only if the
model is flat.

We are ready to state the main result for the constant solutions of RCM.

\begin{thm}\label{thm:uniq_stat_sol}
The RCM admits a constant componentwise solution $u$, which for all
$p\geq1$ lies in $W^{s,p}$ if and only if $s<s_0(p)$,
\[
	s_0(p)\coloneqq\frac{1}{3}\biggl(\alpha-\frac{d}{2}\biggr)+\frac{1}{2}\bigl(\ell_{3/2}-\ell_{p/2}\bigr).
\]
This is the unique constant solution inside any $H^s$. It has an
explicit formula given by
	\begin{equation}\label{e:solution_with_q}
	u_j = f\cdot 2^{q|j|+q}\prod_{k\leq j}\sqrt{d_k}
                ,\qquad j\in J,
	\end{equation}
where
	\begin{equation}\label{eq:defq}
		q\coloneqq-\frac13(\alpha+d)-\frac12\ell_{3/2}.
	\end{equation}
	A sufficient condition for the solution to be Leray is
        $\alpha>\frac{d}{2}$, for in that case $s_0(2)>0$ and hence
        $u\in H$.
\end{thm}

To prove this theorem, we will need the following Lemma.

\begin{lem}\label{lem:alg_identity}
        If the model is RCM, then for any real function~$\varphi$ and
	any positive integer $n$, \[ \sum_{|j|=n}\prod_{k\leq
	j}\varphi(d_k)
	=\biggl(\sum_{\omega \in \Omega}\varphi(\delta_\omega)\biggr)^{\!n}.  \]
\end{lem}
\begin{proof} 
	Both sides of the identity are equal to
	\[
	\sum_{z\in\Omega^n}\varphi(\delta_{z_1})\varphi(\delta_{z_2})\dots\varphi(\delta_{z_n}).
	\qedhere
	\]
\end{proof}

\begin{proof}[Proof of Theorem~\ref{thm:uniq_stat_sol}]
        Since the model has repeated coefficients, we can look for a fixed
        point of recursion~\eqref{e:recur_qj},
        \[
        q
        =-\frac12\alpha-\frac12\log_2\biggl(\sum_{\omega\in\Omega}\delta_\omega^{3/2}2^q\biggr)
        =-\frac{1}{2}\alpha -\frac{1}{2}q-\frac{1}{2}\log_2\sum_{\omega\in \Omega}\delta_\omega^{3/2},
	\]
        which can be solved in $q$, yielding~\eqref{eq:defq}, thanks
        to the definition of $\ell_{3/2}$ in~\eqref{e:def_ell}.

        If we consider $q_j\equiv q$ and write the corresponding $u_j$ as in~\eqref{e:u_from_q}, we obtain~\eqref{e:solution_with_q}, and since $q$
        solves~\eqref{e:recur_qj}, then $(u_j)_{j\in J}$ is a constant
        componentwise solution. Uniqueness will follow from
        Theorem~\ref{t:atmostonehr} if we can prove that $u\in H^s=W^{s,2}$
        for $s<s_0(2)$.

	To show that $u\in W^{s,p}$ if and only if $s<s_0(p)$, we can apply
        Lemma~\ref{lem:alg_identity}, together with the definitions of
        $\ell_{p/2}$, $q$ and $s_0$, to compute

	\begin{equation*}
        \begin{split}
	        \|u\|_{W^{s,p}}^p & =\sum_{j\in J}2^{ps|j|}2^{d(\frac p2-1)|j|}u_j^p
			=f^p2^{pq}\sum_{n=0}^\infty 2^{[p(q+s)+d(\frac p2-1)]n}\sum_{|j|=n}\prod_{k\leq j}d_k^{p/2}\\
			& =f^p2^{pq}\sum_{n=0}^\infty 2^{[p(q+s)+d(\frac p2-1)]n}\biggl(\sum_{\omega \in \Omega}\delta_\omega^{p/2}\biggr)^{\!n}\\
			& =f^p2^{pq}\sum_{n=0}^\infty 2^{[p(q+s)+d(\frac p2-1)+\frac p2\ell_{p/2}+d]n}
			=f^p2^{pq}\sum_{n=0}^\infty 2^{p(s-s_0)n}.\qedhere
        \end{split}
    \end{equation*}
\end{proof}

\begin{remark}
        The constant solution of the RCM turned out to be what is
        usually called a \emph{binomial cascade}, but with
        \emph{deterministic} multipliers $w_k:=2^qd_k^{1/2}$.  In
        fact, in today's physical models, the multipliers of the
        wavelet coefficients are usually chosen to be i.i.d.~random
        variables (see
        again~\cite{benzi1993random,arneodo1998random}), but our
        solution does not exactly belong to this class, since the
        weights are deterministic and moreover there is the
        constraint
        \[
        \sum_{k\in\mathcal O_j}w_k^3=2^{-\alpha},
        \]
        which follows from~\eqref{eq:defq} and rules out independence.
        Models like this were studied for example by Meneveau and
        Sreenivasan in~\cite{meneveau1991multifractal} and in the
        seminal work by Eggleston~\cite{eggleston1949fractional}.
\end{remark}

\begin{remark}
        Notice that the $H^s$-regularity of the solution from
        Theorem~\ref{t:existence_hr} is much lower than what
        Theorem~\ref{thm:uniq_stat_sol} says. In fact the former was
        far from sharp in its generality, while the latter gives
        optimal regularity for RCM.

	For RCM we also have a closed form for the energy of the
        constant Leray solution, when $s_0(2)>0$: \[ \sum_{j\in
        J}u_j^2 =\frac{f^22^{2q}}{1-2^{-2s_0(2)}}.  \]
\end{remark}

\begin{remark}
	Lemma~\ref{lem:alg_identity} and
        Theorem~\ref{thm:uniq_stat_sol} may be generalized from RCM to
        the case where the set of the prescribed coefficients is fixed
        within each generation, but can change from one generation to
        the next one. This is not as general a case as the one
        considered in Theorems~\ref{t:existence_hr}
        and~\ref{t:atmostonehr}, but it still extends quite a lot the
        possible choices of coefficients.
\end{remark}

\section{Structure function}\label{s:structure_function}
In this section we prove some properties of the structure function for
the constant Leray solution of the RCM.  In particular we are
interested in comparing its behaviour with the Kolmogorov K41 law.

We work on the abstract model and hence, by virtue of the
considerations in Section~\ref{s:physical_space}, we may
take~\eqref{e:defzetap} as the definition of $\zeta_p$ for a constant
componentwise solution $(u_j)_{j\in J}$ of the abstract model.

We recall that $\zeta_p$ is then interpreted as the exponent of the
structure function for the reconstructed ``physical'' solution
$u(x)=\sum_{j\in J}u_j\psi_j(x)$.

\begin{thm}\label{t:Cs_and_zetap_RCM}
Consider an RCM. We introduce the quantity
\[
h=\frac13\biggl(\alpha-\frac d2\biggr)-\frac12(\ell_\infty-\ell_{3/2}).
\]
Suppose $h\in(0,1)$. Then there exists a unique constant Leray
solution which lies in $C^s$ if and only if $s\leq h$ and for which the exponents
$\zeta_p$ of the structure function are given by
\begin{equation}\label{eq:zetap_explicit}
  \zeta_p=\min\biggl\{p;\quad\frac{p}{3}\biggl(\alpha-\frac{d}{2}\biggr)+\frac{p}{2} (\ell_{3/2}-\ell_{p/2})\biggr\}
,\qquad p\geq0.
\end{equation}
This function is continuous, non-decreasing, concave, satisfies
$\zeta_0=0$ and $\zeta_3=\min\{3;\alpha-d/2\}$, has oblique asymptote
  of equation $hp+d-\log_2m$, where $m$ is the multiplicity of the
  largest $\delta_\omega$.
\end{thm}

It is interesting to notice that when $\alpha=\frac d2+1$,
\[
\zeta_0= 0 \qquad \zeta_3=1,
\]
since these are physical requirements of turbulence theory and
$\alpha=\frac d2+1$ is the physically meaningful value. In particular
the second one arises from the (non-phenomenological) Kolmogorov
four-fifths law, as shown for example by Frisch
in~\cite{MR1428905}. With the same parameters the theorem also states
that the constant solution has H\"older regularity $h<1/3$ (unless the
model is flat), so the constant solution is one example of what the
second half of Onsager conjecture suggests. See also
Remark~\ref{rem:Onsager_local} below for more on this matter.
\begin{proof}
By the definition of $\zeta_p$ given in equation~\eqref{e:defzetap}
we need to show that
\[
\lim_{n\to\infty}\frac1n\log_2\sum_{|j|=n}u_j^p
=-p\frac\alpha3-p\frac d3-\frac p2\ell_{3/2}+\frac p2\ell_{p/2}+d.
\]
By Lemma~\ref{lem:alg_identity} and equation~\eqref{e:solution_with_q},
\[
\sum_{|j|=n}u_j^p
=f^p2^{pqn+pq}\sum_{|j|=n}\prod_{k\leq j}d_k^{p/2}
=f^p2^{pqn+pq}\biggl(\sum_{\omega\in\Omega}\delta_\omega^{p/2}\biggr)^n,
\]
so by equation~\eqref{eq:defq},
\[
\lim_{n\to\infty}\frac1n\log_2\sum_{|j|=n}u_j^p
=pq+\frac p2\ell_{p/2}+d
=-p\biggl(\frac\alpha3+\frac d3+\frac12\ell_{3/2}\biggr)+\frac p2\ell_{p/2}+d,
\]
as claimed.

We can check that continuity of $\zeta_p$ is a consequence of that of
$\ell_p$. Concavity follows from convexity of $p\ell_p$ which can be
proven by combining the definition and Jensen inequality: let
$\theta\in[0,1]$, then
\begin{multline*}
\log\sum_\omega\delta_\omega^{\theta p+(1-\theta)q}-\theta\log\sum_\omega\delta_\omega^p-(1-\theta)\log\sum_\omega\delta_\omega^q\\
=\log\sum_\omega\biggl(\frac{\delta_\omega^p}{\sum_i\delta_i^p}\frac{\sum_j\delta_j^q}{\delta_\omega^q}\biggr)^\theta\frac{\delta_\omega^q}{\sum_k\delta_k^q}
\leq\log\biggl(\sum_\omega\frac{\delta_\omega^p}{\sum_i\delta_i^p}\biggr)^\theta
=0.
\end{multline*}
The limit of $\zeta_p/p$ as $p\to\infty$ is $h$ which is non-negative
by hypothesis, so monotonicity comes as a consequence of
concavity. The asymptote is an easy limit:
\[
\zeta_p-hp
=\frac p2(\ell_\infty-\ell_{p/2})
=-\log_2\biggl[\frac1N\sum_\omega\biggl(\frac{\delta_\omega}{\delta_{\max}}\biggr)^{p/2}\biggr],
\]
which converges to $d-\log_2m$ as $p\to\infty$.

As for H\"older regularity, by equation~\eqref{e:solution_with_q},
\[
\max_{|j|=n}u_j
=f2^{qn+q}\delta_{\max{}}^{n/2},
\]
so $ns+\frac12dn+\max_{|j|=n}\log_2u_j$ is bounded in $n$ if and only if
\[
s+d/2+q+\ell_\infty/2\leq0.
\]
Substituting $q$ by~\eqref{eq:defq} we get $s\leq h$.
\end{proof}

\begin{remark}
The first derivative of $\zeta_p$ is,
\[
\zeta_p'
=\frac13\biggl(\alpha-\frac d2\biggr)+\frac12\ell_{3/2}-\frac12\sum_\omega\frac{\delta_\omega^{p/2}}{\sum_i\delta_i^{p/2}}\log_2\delta_\omega,
\]
which for $p=0$ reduces to $\zeta_0'=\frac13(\alpha-\frac
d2)+\frac12\ell_{3/2}-\frac12\ell_0$. If this quantity is 1 or less,
then $p$ is never the minimum in equation~\eqref{eq:zetap_explicit},
thus $\zeta_p$ is strictly concave and smooth for all $p$. On the
other hand, if the derivative in 0 is larger than 1, then since $h<1$,
there exists $p_0>0$ such that $\zeta_p=p$ if and only if $p\leq p_0$.
\end{remark}

\begin{remark}
The condition $h>0$ is fundamental. If $h<0$ the right-hand side
of~\eqref{eq:zetap_explicit} is decreasing and then negative for large
$p$ and the arguments of Section~\ref{s:physical_model} are no longer
valid when $\zeta_p<0$, so we do not know how to compute the exponents
of the structure function for those values of $p$. If $h=0$
equation~\eqref{eq:zetap_explicit} holds, but $C^0$ is not defined.

The condition $h\leq1$ could be weakened, but it is very reasonable,
since $h\leq\frac13(\alpha-d/2)$ and usually $\alpha=1+d/2$.
\end{remark}

\subsection{Comparison to other models}
As we mentioned in the introduction, several models were suggested in
the literature for which the function $\zeta_p$ can be computed, and
there are also experimental data available, so we want to compare our
function to both.

The first model was given by Kolmogorov in~\cite{K41}, as a uniform
cascade of energy and it simply yields the line $\zeta_p=\frac{p}{3}$.
A different solution, trying to cope with the intermittency observed
in experimental data, led twenty years later to the development by
Obhukov and Kolmogorov of the
\emph{log-normal} model~\cite{obukhov1962some,K62}, which yields:
\[
	\zeta_p=\frac{p}{3}+\frac{\mu}{18}(3p-p^2).
\] 
However this model has the big drawback of being eventually
decreasing, which allows for supersonic velocities, as well as some
other issues. Nevertheless, it paved the way for subsequent models.

The $\beta$-model was introduced by Frisch et
al.~\cite{frisch1978simple} as a toy model to investigate some of the
fractal properties of turbulence, as suggested by Mandelbrot in
several papers, for
example~\cite{mandelbrot1974intermittent}. This
model is fractal by construction, but turns out to be monofractal,
again with a linear $\zeta_p$:
\[
	\zeta_p=\frac{p}{3}+(3-D)\biggl(1-\frac{p}{3}\biggr).
\]
This model was then generalized to a bifractal model, which is just a
mixture of two different $\beta$-models, combining into a piecewise
linear map, with one change of slope.

After the experimental results of Anselmet et
al.~\cite{anselmet1984high} became available, Frisch and
Parisi~\cite{parisi1985singularity} made the crucial remark that
$\zeta_p$ could be seen in a multifractal framework, one possible
example being the random
$\beta$-model~in~\cite{benzi1984multifractal}. Many more examples
followed, thanks to the vitality of the multifractal community.

Finally, She and L\'ev\^eque introduced in~\cite{she1994universal} a
phenomenological model based on fluctuation structures associated with
vortex filaments; it is free of parameters and has a good fit to
experimental data:
\[
	\zeta_p=\frac{p}{9}+2-2\biggl(\frac{2}{3}\biggr)^{p/3}.
\]

\begin{figure}[ht!]
	\centering
	\includegraphics[width=0.995\textwidth]{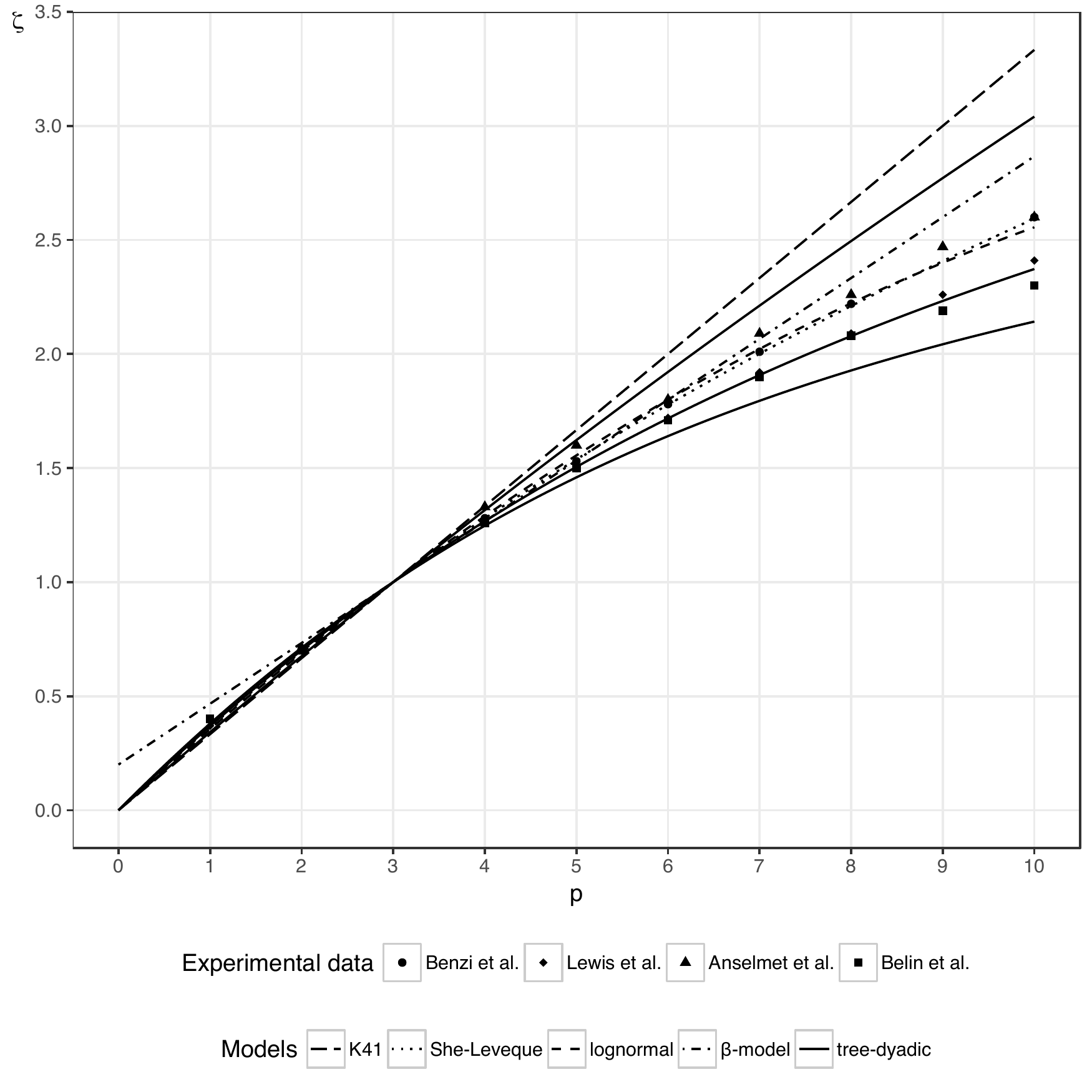}
	\caption{Comparison of $\zeta_p$ functions from different models and experimental data.}
	\label{fig:zp_plot}
\end{figure}

In Figure~\ref{fig:zp_plot} we show the plots of the functions
$\zeta_p$ for our model, with three different choices of parameters, and
some of the other ones cited here, as well as experimental data from
Anselmet et al.~\cite{anselmet1984high}, Belin et
al.~\cite{belin1996exponents}, Benzi et
al.~\cite{benzi1996generalized} and Lewis et
al.~\cite{PhysRevE.59.5457}.  The choice of parameters is the
following: in the log-normal model $\mu=0.2$, in the $\beta$ model
$D=2.8$, in the tree-dyadic
model
\[
(\log_2\delta_\omega)_{\omega}=(\lambda i, i=0,1,\dots,7)
\]
with $\lambda=0.1$ for the top one, $\lambda=0.2$ for the middle one
and $\lambda=0.2307$ for the bottom one.

Let us spend a couple more words on the choice of the coefficients for
the RCM. Given the number of degrees of freedom we have in the choice
of these coefficients (which are 7, since multiplicative constants for
$\delta_\omega$ do not count), it is not particularly informative to
show that we can fit precisely the experimental data. It is rather
more interesting to show that, even considering just linear steps in
the logarithms ---to reduce to a family with one degree of freedom--- we
can cover quite a variety of situations.  The extremal cases are
upwards the Kolmogorov line K41, (corresponding to the flat model,
with $\lambda=0$), and downwards $\lambda\approx0.2307$ which is close
to the constraint $h>0$.
\section{Fractality}\label{s:fractality}
In this section we consider again the physical field
$u(x)\coloneqq\sum_ju_j\psi_j(x)$ reconstructed from the constant
solution of the RCM. It is defined on the physical space $Q_\emptyset$
and is multifractal in nature. In particular for every level of
regularity (think of $C^s$ for example) there is a set of points $x$
for which $u$ around $x$ attains that regularity locally.

Anomalous dissipation depends on regularity; in the tradeoff
between low regularity and high Hausdorff dimension of the set, we
look for the critical set which accounts for most anomalous dissipation.

The first proposition computes the energy flow for a finite rooted
subtree and explicits the term which we will identify with anomalous
dissipation.

\begin{prop}\label{p:energy_flow_terms}
Let $T$ be a finite subset of $J$ with the property that $j\in
T\Rightarrow\bar\jmath\in T$, Let $\partial T$ be the set of nodes
outside $T$ with father in $T$ and let $v$ be a componentwise
solution of~\eqref{e:model}. Then
\[
\frac d{dt}\sum_{j\in T}v_j^2(t)
=2f^2v_\emptyset(t)-\sum_{j\in\partial T}2c_jv_{\bar\jmath}(t)^2v_j(t).
\]
\end{prop}

\begin{proof}
Since $T$ is finite we can exchange derivative and sum,
\begin{align*}
\frac d{dt}\sum_{j\in T}v_j^2
&=\sum_{j\in T}2v_jv_j'
=\sum_{j\in T}2v_j\biggl(c_jv_{\bar\jmath}^2-\sum_{k\in O(j)}c_kv_jv_k\biggr)\\
&=\sum_{j\in T}2c_jv_{\bar\jmath}^2v_j-\sum_{k:\bar k\in T}2c_kv_j^2v_k.
\end{align*}
By the hypothesis on $T$ and the definition of $\partial T$, we have
\[
	\{k\in J:\bar k\in T\}\cup\{\emptyset\}=T\cup\partial T. 
\]
Since the
contribution of $\emptyset$ is $2f^2v_\emptyset(t)$, the proof is complete.
\end{proof}

\begin{remark}
The generality of the set $T$ in Proposition~\ref{p:energy_flow_terms}
allows us to give an interpretation of the term
$2c_ju_{\bar\jmath}^2u_j$ as the energy flow from $\bar\jmath$ to
$j$. During each unit of time this amount of energy enters the subtree
rooted in $j$ and distributes among all the subtree's nodes,
contributing to the wavelet components of the solution corresponding
to these nodes. Notice that these components are all supported inside
the cube $Q_j$, and we are considering the constant solution, so the
same amount of energy must be dissipated inside the cube $Q_j$. Thus
the quantity
\begin{equation}\label{e:Fj_def}
F_j
\coloneqq \frac{2c_ju_{\bar\jmath}^2u_j}{2c_\emptyset u_{\bar\emptyset}^2u_\emptyset}
=\frac{1}{2^qf^3}c_ju_{\bar\jmath}^2u_j
,\qquad j\in J,
\end{equation}
can be interpreted as the fraction of anomalous dissipation inside
cube $Q_j$.

Notice moreover that if $T$ is as in
Proposition~\ref{p:energy_flow_terms}, then the family
$(Q_j)_{j\in\partial T}$ forms a partition of $Q_\emptyset$ made of smaller
non-overlapping cubes. In this sense
Proposition~\ref{p:energy_flow_terms} states that for any such
partition of $Q_\emptyset$, the total energy dissipation of the system is the
sum of the anomalous dissipation of every cube of the partition, and
that this sum does not depend on the partition itself and it is always
equal to the energy entering the system from its root.
\end{remark}

The question arises now whether the anomalous dissipation is
distributed somewhat evenly among the cubes of a partition. If this
was the case, it would be more or less proportional to the volume of
the cubes and there would be a density of anomalous dissipation with
respect to the Lebesgue measure $\mathcal L$. This is not the case, as
the following statement clarifies.

\begin{prop}\label{p:rate_Fj}
Let $u$ be the constant solution of an RCM, and $(F_j)_{j\in J}$
defined as in~\eqref{e:Fj_def}. Let
\begin{equation}\label{e:calE_def}
\mathcal{R}(a)\coloneqq d+\frac32\ell_{3/2}-\frac32a
,\qquad a\in\mathbb R.
\end{equation}
Then the following holds:
\begin{enumerate}
\item Anomalous dissipation of energy in the cubes has an exponential rate in
  $|j|$ that can be computed explicitly:
\begin{equation}\label{e:1nlog2Fj}
\frac1{|j|}\log_2F_j=-\mathcal{R}(\sigma_j)
,\qquad j\in J, 
\end{equation}
where we define 
\begin{equation*}
\sigma_j\coloneqq\frac1{|j|}\sum_{k\leq j}\log_2d_k
,\qquad j\in J.
\end{equation*}
\item Introduce the pointwise rate of anomalous dissipation,
\[
\sigma(x)\coloneqq\lim_{n\to\infty}\sigma_{x_n},
\]
for all $x\in Q_\emptyset$ for which the limit exists. Then
$\sigma(x)=\ell_0$ for $\mathcal L$-a.e.~$x$.
\item  For all $x$ such that $\sigma(x)<\ell_{3/2}$,
\[
\lim_{n\to\infty}\frac{F_{x_n}}{\mathcal L(Q_{x_n})}=0.
\]
In particular if the model is not flat, then this holds for $\mathcal
L$-a.e.~$x$.
\end{enumerate}
\end{prop}

\begin{proof}
By substituting the definition~\eqref{e:solution_with_q} inside~\eqref{e:Fj_def}, we get
\[
F_j=2^{\alpha|j|}2^{3q|j|}\prod_{k\leq j}d_k^{3/2}.
\]
We can now recall that, by~\eqref{eq:defq}, 
$q=-\frac13\alpha-\frac13d-\frac12\ell_{3/2}$, so that
\[
\frac1{|j|}\log_2F_j
=\alpha+3q+\frac32\frac1{|j|}\sum_{k\leq j}\log_2d_k
=-d-\frac32\ell_{3/2}+\frac32\sigma_j
=-\mathcal R(\sigma_j).
\]

For the second part, consider the probability space
$(Q_\emptyset,\mathcal B,\mathcal L)$.  The maps
$x\mapsto d_{x_i}$, for $i\in \N$ are random variables, and so is $\sigma_{x_n}$,
\[
\sigma_{x_n}=\frac1n\sum_{i=0}^n\log_2d_{x_i}.
\]
By the definition of RCM, for all $i\in\N$ the law of $d_{x_i}$
conditioned on $d_{x_{i-1}}$ is uniform on the set
$\{\delta_\omega\}_{\omega\in \Omega}$, hence the random process
$(d_{x_i})_{i\in \N}$ is a sequence of i.i.d.~random variables.  By
the strong law of large numbers,
\[
\sigma_{x_n}\to\frac1N\sum_\omega\log_2\delta_\omega
\qquad\mathcal L\text{-a.e. }x,\quad\text{as }n\to\infty.
\]
By the definition of $\ell_0$ this completes the second part. As for the
last part,
\[
\frac1n\log_2\frac{F_{x_n}}{\mathcal L(Q_{x_n})}
=\frac1n\log_2F_{x_n}+d, 
\]
and the right-hand side converges almost surely to $-\frac32(\ell_{3/2}-\sigma(x))$ as $n\to\infty$.

The hypothesis that the model is not flat ensures that
$\ell_{3/2}>\ell_0$.
\end{proof}

Proposition~\ref{p:rate_Fj} states, in the first point, that the
anomalous dissipation of cube $Q_j$ depends on $\sigma_j$. In
particular if the anomalous dissipation was evenly distributed, $F_j$
would be proportional to the volume $2^{-dn}$ and hence
by~\eqref{e:1nlog2Fj} the typical value of $\sigma_j$ would be
$\ell_{3/2}$. On the contrary, the second point in
Proposition~\ref{p:rate_Fj} states that the typical value is $\ell_0$
instead, which is lesser, and cannot account for a positive fraction
of the total anomalous dissipation (hence the 0 density limit).  This
means that anomalous dissipation is actually concentrated in few cubes
with much larger values of $\sigma_j$ and $F_j$.  This in turn
suggests that we are dealing with a fractal object, and in particular
that Lebesgue measure is not the right mathematical tool to get a
meaningful picture of this phenomenon.

\begin{remark}\label{rem:Onsager_local}
From a local point of view, Proposition~\ref{p:rate_Fj} further
clarifies that pointwise anomalous dissipation happens exactly at the
points $x$ such that $\sigma(x)\geq\ell_{3/2}$. This can be linked to
some sort of local H\"older exponent, in fact the description of the
spaces $C^s$ in terms of wavelet coefficients given in
Definition~\ref{d:function_spaces} suggests a pointwise refinement by
introducing the local H\"older exponent of $u$ at the point $x$ as
\[
s(x):=\sup\left\{s:\sup_{n\geq1}\biggl(ns+\frac12dn+\log_2|u_{x_n}|\biggr)<\infty\right\},
\]
or equivalently
\[
s(x)=\liminf_{n\to\infty}\left(-\frac d2-\frac1n\log_2|u_{x_n}|\right).
\]
Actually, there exists a different commonly accepted definition of
local H\"older exponent: our $s(x)$ is in principle a different
quantity (also found in the literature, often called the local
singularity exponent of wavelet coefficients and denoted by
$w(x)$). In many simple cases these two concepts are equivalent, but
not in general, as is shown in Muzy et
al.~\cite{muzy1993multifractal}. (We refer the reader to
Riedi~\cite{riedi1999multifractal} for more details.)

In the case of the constant solution of the RCM we get
\[
s(x)=\frac13\biggl(\alpha-\frac d2\biggr)+\frac12(\ell_{3/2}-\sigma(x)).
\]
Then for the physical case, when $\alpha=1+d/2$, we get that
$\sigma(x)\geq\ell_{3/2}$ if and only if $s(x)\leq\frac13$: there is
anomalous dissipation at a point $x$ if and only if $s(x)\leq\frac13$.
Notice that the ``only if'' part of this pointwise statement also holds for
incompressible Euler equations, as was first shown by Duchon and
Robert~\cite{ducrob00}.
\end{remark}

The following theorem, which could be restated in terms of a large
deviation principle for $\sigma_j$, identifies exactly the single
value of $\sigma_j$ which contributes to almost all the anomalous
dissipation.

It will be useful to introduce the following function:
\begin{equation}\label{e:phi_def}
\phi(\gamma)\coloneqq \sum_\omega\frac{\delta_\omega^\gamma}{\sum_v\delta_v^\gamma}\log_2\delta_\omega
,\qquad\gamma\in\mathbb R,
\end{equation}
which we notice satisfies
\begin{equation}\label{eq:derivative}
\frac {d}{d\gamma}(\gamma\ell_\gamma)
=\phi(\gamma)
,\qquad\gamma\in\mathbb R.
\end{equation}

\begin{thm}\label{t:AD_is_phi32}
For all sets $B\subset\mathbb R$ for which $\phi(3/2)$ is an internal
point,
\[
\lim_{n\to\infty}\sum_{|j|=n}F_j\mathbbm{1}_{\sigma_j\in B}=1.
\]
\end{thm}

\begin{remark}\label{rem:post_thm_AD_is_phi32}
Let $S_a\coloneqq\{x:\sigma(x)=a\}$. In non-rigorous terms,
Theorem~\ref{t:AD_is_phi32} states that the set $S_{\varphi(3/2)}$
accounts for all anomalous dissipation.  Notice that
$\varphi(3/2)\geq\ell_{3/2}$, as can be deduced by
equation~\eqref{e:calD_def} below, so considering $S_a$ for increasing
values of $a$, we get the picture that anomalous dissipation starts
when $a=\ell_{3/2}$ and increases in intensity with $a$. When
$a=\varphi(3/2)$ the tradeoff between intensity of anomalous
dissipation and Hausdorff dimension of the set $S_a$ balances out and
we may say that all anomalous dissipation happens in
$S_{\varphi(3/2)}$.
\end{remark}

To prove Theorem~\ref{t:AD_is_phi32} we will need a couple of
technical results.

\begin{lem}\label{l:maxH_on_symplex}
Let $S$ be the canonical simplex of $\mathbb R^\Omega$,
\[
	S\coloneqq \{p\in\mathbb R_+^\Omega:\sum_{\omega\in\Omega}p_\omega=1\}.
\]
Let $H$ be the entropy and $\sigma$ a linear function on $S$,
\[
H(p)\coloneqq -\sum_{\omega\in\Omega}p_\omega\log_2p_\omega,
\qquad\qquad
\sigma(p)\coloneqq \sum_{\omega\in\Omega}p_\omega\log_2\delta_\omega.
\]
Suppose $\ell_{-\infty}\neq\ell_\infty$, then the map $\phi$ defined by
equation~\eqref{e:phi_def} is a strictly increasing bijection from
$\mathbb R$ to $(\ell_{-\infty},\ell_\infty)$. For all
$a\in(\ell_{-\infty},\ell_\infty)$ let $\gamma_a\coloneqq \phi^{-1}(a)$.
Then the maximum value of $H$ on $S$ subject to the constraint
$\sigma(p)=a$ is
\begin{equation}\label{e:calD_def}
	\mathcal D(a)
	\coloneqq d-\gamma_a(a-\ell_{\gamma_a})
	=\max_{\substack{p\in S\\\sigma(p)=a}}H(p)
        \leq d.
\end{equation}

Otherwise, if  $\ell_{-\infty}=\ell_\infty=:l$, then $\sigma(p)\equiv l$ is constant and
\[
\max_{\substack{p\in S\\\sigma(p)=l}}H(p)
=\max_{p\in S}H(p)
=d
\eqqcolon\mathcal D(l).
\]
\end{lem}
\begin{remark}
	Notice that $\mathcal D$ is defined differently in the two cases, but
the two definitions are at least compatible, in the sense that in
both cases $\mathcal D(\ell_0)=d$.
\end{remark}

\begin{figure}[ht!]
	\centering
	\includegraphics[width=0.995\textwidth]{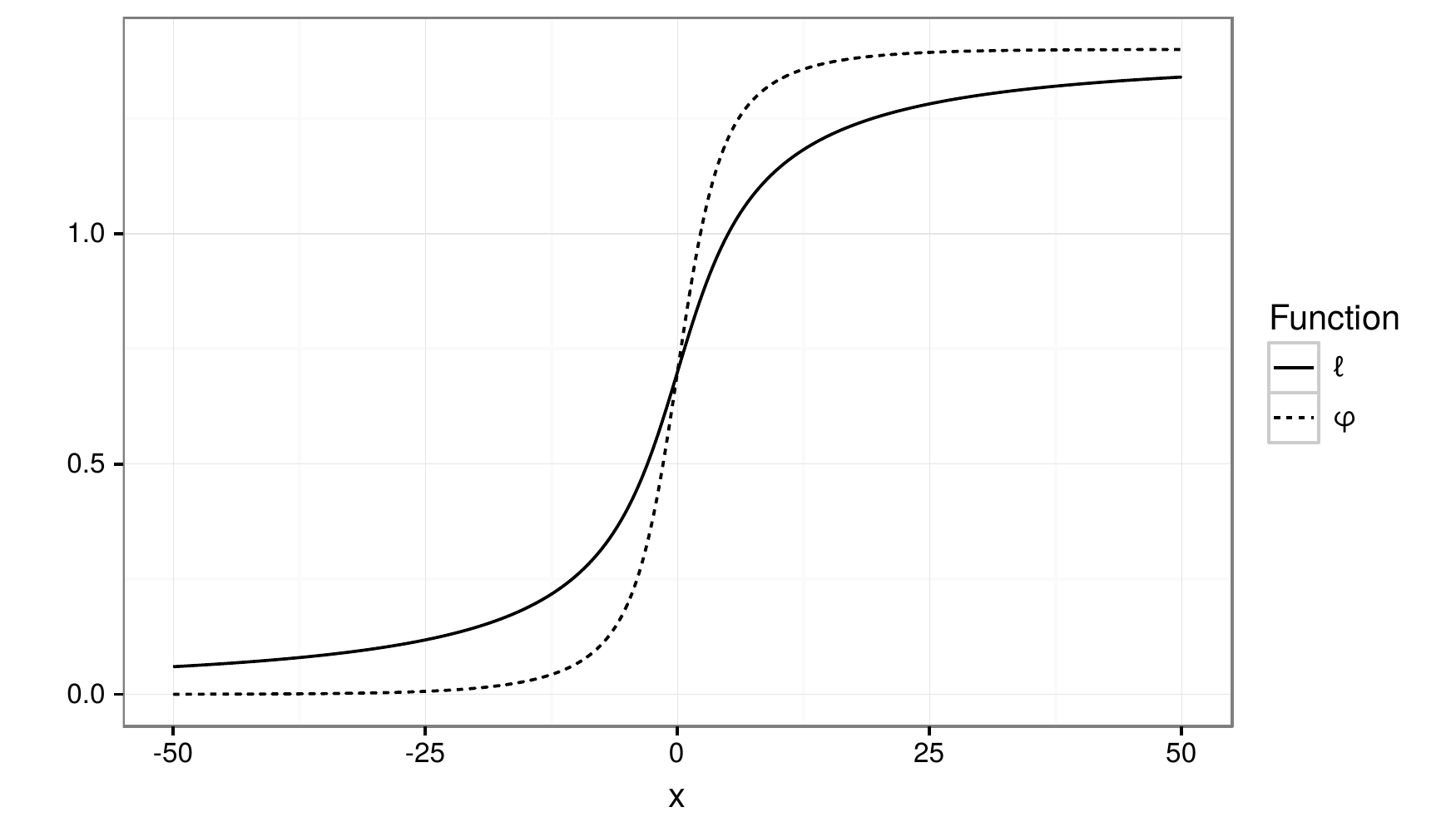}
	\caption{Comparison of the functions $\ell$ and $\varphi$ for a given choice of coefficients $(\delta_\omega)_{\omega\in\Omega}$.}
\end{figure}

\begin{proof}[Proof of Lemma~\ref{l:maxH_on_symplex}]
If $\ell_{-\infty}=\ell_\infty$, then the model is flat, the
$\delta_\omega$'s are all equal to $\delta=2^l$ and the constraint
$\sigma(p)=l$ becomes trivially true. In that case $\mathcal D(a)$ is
defined only for $a=l$ and equal to $d$, which is exactly the maximum
of entropy under the single constraint of satisfying the simplex
equation.

From now on we will suppose that the model is not flat. By the method
of Lagrange multipliers applied to $H$ with two constraints given by
$\sigma(p)=a$ and the simplex equation, we can immediately get that
for any stationary point~$\widehat p$,
\[
\widehat p_\omega
=c\delta_\omega^\gamma, 
\]
for suitable constants $c$ and $\gamma$. From the simplex condition we have
$c^{-1}=\sum_\omega\delta_\omega^\gamma$. From the other constraint we obtain
\[
a
=\sigma(\widehat p)
=\sum_{\omega\in\Omega}\widehat p_\omega\log_2\delta_\omega
=\sum_{\omega\in\Omega}\frac{\delta_\omega^\gamma}{\sum_v\delta_v^\gamma}\log_2\delta_\omega
=\phi(\gamma).
\]
The derivative of $\phi$ is non-negative, since it can be expressed as
the variance of a discrete random variable:
\[
\phi'(\gamma)
=\sum_\omega\frac{\delta_\omega^\gamma}{\sum_v\delta_v^\gamma}(\log_2\delta_\omega)^2-\biggl(\sum_\omega\frac{\delta_\omega^\gamma}{\sum_v\delta_v^\gamma}\log_2\delta_\omega\biggr)^2
\geq0.
\]
In particular $\phi'(\gamma)\neq0$ since $\delta_\omega$ are not all
equal and hence $\phi$ is a bijection from $\mathbb R$ to
$(\ell_{-\infty},\ell_\infty)$.

We can thus invert $a=\phi(\gamma)$, find
$\gamma=\gamma_a=\phi^{-1}(a)$ and compute $H(\hat p)$
\begin{equation*}
\begin{split}
	H(\widehat p) &=-\sum_{\omega\in\Omega}\frac{\delta_\omega^\gamma}{\sum_v\delta_v^\gamma}\log_2\frac{\delta_\omega^\gamma}{\sum_v\delta_v^\gamma}\\ 
	&=\log_2\sum_v\delta_v^\gamma-\sum_{\omega\in\Omega}\frac{\delta_\omega^\gamma}{\sum_v\delta_v^\gamma}\log_2\delta_\omega^\gamma
	=d+\gamma\ell_\gamma-\gamma\phi(\gamma)
	=\mathcal D(a).
\end{split}
\end{equation*}
To conclude it is enough to notice that $H$ is concave, since its
Hessian matrix is diagonal negative definite.
\end{proof}

\begin{lem}\label{l:calE_geq_calD}
Consider the functions $\mathcal{R}$ and $\mathcal D$ as defined in
equations~\eqref{e:calE_def} and~\eqref{e:calD_def}. The following inequality holds: 
\[
\mathcal{R}(a)\geq\mathcal D(a)
,\qquad a\in(\ell_{-\infty},\ell_\infty), 
\]
with equality if and only if $a=\phi(3/2)$.
\end{lem}

\begin{proof}
Let us consider the difference $\mathcal{R}(\phi(\gamma))-\mathcal
D(\phi(\gamma))$ as a function of $\gamma$. We want to prove that
\[
\frac{3}{2}\bigl(\ell_{3/2}-\phi(\gamma)\bigr)
-\gamma\bigl(\ell_\gamma-\phi(\gamma)\bigr)
\geq0,
\]
with equality if and only if $\gamma=\frac32$. The \emph{if} part of the equality case
is obvious, while the strict inequality for $\gamma\neq\frac{3}{2}$ comes
by Taylor formula for the function $s\mapsto s\ell_s$ in~$\gamma$. 

As we noticed in~\eqref{eq:derivative}, we have for all $s$,
\[
\frac {d}{ds}(s\ell_s)
=\phi(s),
\]
so we can write, for a suitable $\xi=\xi(s)\in(\gamma,s)$, 
\[
s\ell_s=\gamma\ell_\gamma+(s-\gamma)\phi(\gamma)+\frac12(s-\gamma)^2\phi'(\xi).
\]
We proved in
Lemma~\ref{l:maxH_on_symplex} that $\phi$ is strictly increasing, so
we get
\[
s(\ell_s-\phi(\gamma))-\gamma(\ell_\gamma-\phi(\gamma))>0, 
\]
for all $s\neq\gamma$.
\end{proof}

We can now proceed with the proof of the theorem.
\begin{proof}[Proof of Theorem~\ref{t:AD_is_phi32}]
  Let $n\geq1$. Since $\sum_{|j|=n}F_j=1$ by the definition of $F_j$,
  the following defines a discrete probability measure on $\R$:
\[
\mu_n\coloneqq \sum_{|j|=n}F_j\delta_{\sigma_j}.
\]
Let $A$ be the complement of $B$ in $\R$.  Having the result of
Lemma~\ref{l:calE_geq_calD} in mind, we will show that
\begin{equation}\label{e:claim_limsup_rate_Tn}
	\limsup_{n\to\infty}\frac1n\log_2\mu_n(A)
	\leq-\inf_{a\in A}\bigl[\mathcal{R}(a)-\mathcal D(a)\bigr].
\end{equation}
Assuming this to hold, by Lemma~\ref{l:calE_geq_calD} and the hypothesis on $B$, namely that $\phi(3/2)$ is an internal point, 
we will get
\[
\inf_{a\in A}\bigl[\mathcal{R}(a)-\mathcal D(a)\bigr]\eqqcolon \lambda>0, 
\]
and hence
\[
\mu_n(B)
=1-\mu_n(A)
\geq1-2^{-\lambda' n},
\]
for $n$ large and a suitable $\lambda'>0$, yielding the desired
conclusion that $\mu_n(B)\to1$ as $n\to\infty$.

To prove~\eqref{e:claim_limsup_rate_Tn}, we use
Proposition~\ref{p:rate_Fj} to rewrite $\mu_n(A)$ in terms of the
$\sigma_j$'s as
\[
	\mu_n(A)\coloneqq \sum_{|j|=n}2^{n\bigl(-d-\frac32\ell_{3/2}+\frac32\sigma_j\bigr)}\delta_{\sigma_j}(A).
\]
Notice that $\sigma_j=\sigma_{j'}$ if $j$ and $j'$ have the same
generation and the $d_k$'s appear the same number of times but in
different order in the definition of $\sigma_j$. This suggests the
change of variables $p=\pi(j)$, where $\pi:J\to\mathbb R^\Omega$ is
defined by
\[
	\pi_\omega(j)
	\coloneqq \frac1{|j|}\sharp\{k\leq j:d_k=\delta_\omega\}
	,\qquad\omega\in\Omega,\quad j\in J.
\]
In fact $\sigma_j$ depends only on $\pi(j)$, and indeed we can write
$\sigma_j=\sigma(\pi(j))$, with $\sigma:\mathbb R^\Omega\to\mathbb R$
defined by
\[
	\sigma(p)
	\coloneqq \sum_{\omega\in\Omega}p_\omega\log_2\delta_\omega
	,\qquad p\in\mathbb R^\Omega.
\]
Now we can rewrite $\mu_n(A)$, with the change of variable $p=\pi(j)$, as
\[
	\mu_n(A)
	=\sum_{p\in S_n}\mathbbm1_{\sigma(p)\in A}2^{n(-d-\frac32\ell_{3/2}+\frac32\sigma(p))}c_n(p),
\]
where 
\[
	c_n(p)=\sharp\{j\in J:|j|=n,\pi(j)=p\},
\]
and $S_n$ is the $\frac1n$-lattice inside the canonical symplex
of $\mathbb R^\Omega$
\begin{equation*}
\begin{split}
	S_n
	&=\pi(\{j\in J:|j|=n\})\\
	&=\{p\in\mathbb R^\Omega:\sum_{v\in\Omega}p_v=1,\textup{ and for all }\omega\in\Omega, p_w\geq0, p_wn\in \Z\}.
\end{split}
\end{equation*}

We want an upper bound for $\mu_n(A)$. The factor $c_n(p)$ can be
computed exactly, as it is easy to see that
\[
	c_n(p)
	=\binom n{p_1n\ p_2n\ \dots\ p_Nn},
\]
and this multinomial can be bounded by one version\footnote{The usual
  Stirling's approximation states that $n!n^{-n-1/2}e^n\to\sqrt{2\pi}$
  as $n\to\infty$. One can also prove that
  $n!n^{-n-1/2}e^n\in[\sqrt{2\pi},e]$ for all $n$.} of Stirling's
approximation, yielding
\[
	\frac1n\log_2c_n(p)
	\leq\frac1n\log_2\left(\frac{n^{(1-N)/2}e}{(2\pi)^{N/2}}\prod_\omega p_\omega^{-p_\omega n+1/2}\right)
	\leq H(p)+C\frac{\log_2n}n,
\]
where $H$ denotes the entropy, defined as 
\[
	H(p)=-\sum_\omega p_\omega\log_2p_\omega,
\] 
and
the constant $C$ does not depend on $p$ or $n$.

The sum over $S_n$ is then bounded by the cardinality $\sharp S_n$
times the supremum of the summand in $p$. We have
\[
	\sharp S_n
	=\binom{n+N-1}{N-1}\leq n^N,
\]
hence $\frac1n\log_2(\sharp S_n) \leq N\frac{\log_2n}n$ and we get
\begin{equation}\label{e:ineqTn}
	\frac1n\log_2\mu_n(A)
	\leq (N+C)\frac{\log_2n}n+\sup_{p\in\sigma^{-1}(A)}\left(-d-\frac32\ell_{3/2}+\frac32\sigma(p)+H(p)\right).
\end{equation}
By Lemma~\ref{l:maxH_on_symplex},
$\sup_{p\in\sigma^{-1}(a)}H(p)=\mathcal D(a)$, so taking the limsup in\eqref{e:ineqTn},
we get
\[
\limsup_{n\to\infty}\frac1n\log_2\mu_n(A)
\leq\sup_{a\in A}\left(-d-\frac32\ell_{3/2}+\frac32a+\mathcal D(a)\right),
\]
which is a rewriting of~\eqref{e:claim_limsup_rate_Tn}.
\end{proof}

Finally, we deal with the Hausdorff dimension of the set of points
that accounts for all anomalous dissipation. We will need to be more
precise than we were in Remark~\ref{rem:post_thm_AD_is_phi32}. There
we defined $S_a\coloneqq\{x:\exists\lim_n\sigma_{x_n}=a\}$. This will
be now refined to $E(\mathbb{S}_a)$, the set of $x$ for which all the
points of accumulation of the relative densities of the
$\delta_\omega$ appearing in the sequence $d_{x_n}$ correspond to
$\sigma=a$.  This notation allows us to use a theorem in
Olsen~\cite{olsen2004hausdorff} to compute the Hausdorff dimension of
$E(\mathbb{S}_a)$.

\medskip  Consider once more the notation introduced in the proof of
Theorem~\ref{t:AD_is_phi32}: the maps $\pi:J\to\mathbb R^\Omega$,
\[
\pi_\omega(j)
= \frac1{|j|}\sharp\{k\leq j:d_k=\delta_\omega\}
,\qquad\omega\in\Omega,\quad j\in J,
\]
 and $\sigma:\mathbb R^\Omega\to\mathbb R$,
\[
\sigma(p)
= \sum_{\omega\in\Omega}p_\omega\log_2\delta_\omega
,\qquad p\in\mathbb R^\Omega.
\]
Consider moreover for $x\in Q_\emptyset$ the set of points of accumulation
of the (vectorial) frequencies of the coefficients
$(\delta_\omega)_{\omega \in \Omega}$ in the dyadic expansion in~$x$:
\[
A(x)
\coloneqq \acc\Bigl[(\pi(x_n))_{n\geq 0}\Bigr]
\subseteq S.
\]
Let us also define
\[
	\mathbb{S}_a\coloneqq \{p\in S, \sigma(p)=a\},
\]
and finally
\[
  E(\mathbb{S}_a)\coloneqq \{x\in Q_\emptyset: A(x)\subseteq\mathbb{S}_a \},
\]
the set of all points $x$ in the cube $Q_\emptyset$ such that the asymptotic frequencies of the $(\delta_\omega)_{\omega\in\Omega}$ associated to~$x$ are in $\mathbb{S}_a$.

With the notation introduced above, the following theorem was proved by Olsen
(see~\cite{olsen2004hausdorff})
\begin{thm}[Olsen]
	The Hausdorff dimension of $E(\mathbb{S}_a)$ is:
	\begin{equation*}
	\dim E(\mathbb{S}_a)= \sup_{p\in \mathbb{S}_a}H(p).
	\end{equation*}
\end{thm}

Thanks to Lemma~\ref{l:maxH_on_symplex}, we can compute this dimension
for all $a$, and in particular, by Theorem~\ref{t:AD_is_phi32} we
obtain the following statement.
\begin{thm}
	For all $a\in [\ell_{-\infty}, \ell_{+\infty}]$, the Hausdorff dimension of the set $E(\mathbb{S}_a)$ is  $\mathcal{D}(a)$. In particular the Hausdorff dimension of the set of the points $x$ where anomalous dissipation occurs is 
	\[
		\Delta=d-\frac{3}{2}\bigl(\phi(3/2)-\ell_{3/2}\bigr).
	\]
\end{thm}

\begin{remark}
	It is worth noting that the value of $\Delta$ is in agreement
        with what was expected in the framework of the Frisch-Parisi
        multifractal model~\cite{parisi1985singularity}, that is
        \[
        \Delta=3\zeta'_3+d-1
        \]
        Heuristically, the multifractal formalism relates the
        Hausdorff dimension $d(h)$ of the sets of points of given
        H\"older exponent $h$ with $\zeta_p$ through a Legendre
        transform:
	\[
	\zeta_p = \min_h(ph-d(h))+d
        ,\qquad
        d-d(h)=\max_p(\zeta_p-ph).
	\]
        If a point $x$ has H\"older exponent $h(x)$, then it is
        expected that the measure of energy dissipation has in $x$ a
        singularity exponent $\nu(x)=3h(x)-1+d$ (this can be deduced
        by the formula for $F_j$, with $\alpha=1+d/2$). Let
        $\nu:=3h-1+d$, then
        \[
        d-d(h)\geq\zeta_3-3h=d-\nu
        \]
        for all $h$, with equality only for $h=\zeta_3'$.  Then summing up
        all energy dissipation at points $x$ with $h(x)=h$ we get
        \[
        \mathcal E(h)\leq\limsup_n2^{d(h)n}2^{-\nu n},
        \]
        hence the only contribution is for $h=\zeta'_3$ and so
        $\Delta=d(\zeta'_3)$ as claimed.
\end{remark}

\appendix
\section{Appendix}\label{app:physical}
In this section we propose an heuristic argument to justify
formula~\eqref{e:defzetap} given in Section~\ref{s:physical_space} for
the exponent of the structure function.

Let $(\psi_j)_{j\in J}$ be a family of wavelets such that $\psi_j$ is
essentially supported on the cube $Q_j$ and they are all rescaled and
translated versions one of the other:
\[
\psi_j(x)=2^{d|j|/2}\psi_\emptyset(2^{|j|}x+\theta_j),
\]
for some ``mother wavelet'' $\psi_\emptyset$.
We consider real values $(u_j)_{j\in J}$ and pose $u(x):=\sum_{j\in
J}u_j\psi_j(x)$, for all $x\in Q_\emptyset$, and
define as usual the structure function
\[
S_p(r):=\int_{Q_\emptyset}\left\langle|u(x)-u(y)|^p\right\rangle_ydx,
\]
where $\langle\cdot\rangle_y$ denotes the average on the points $y$
such that $|y-x|=r$, and its exponents,
\[
\zeta_p:=-\lim_{n\to\infty}\frac1n\log_2S_p(2^{-n}).
\]
We introduce also the function
\[
\xi_p
:=d-\frac p2d-\limsup_{n\to\infty}\frac1n\log_2\sum_{|j|=n}|u_j|^p
,\qquad p\geq0.
\]
We want to show that under suitable hypothesis, if $\xi_p>0$, then
$\zeta_p=\min(p,\xi_p)$.

\begin{remark}\label{rem:a5}
For any map $\phi:J\to\R$, for almost every $x\in Q_\emptyset$,
\[
	\sum_{j\in J}\phi(j)\psi_j(x)
	=\sum_{i\geq0}\phi(x_i)\psi_{x_i}(x).
\]
\end{remark}

\begin{lem}\label{lem:6.1}
Let $(a_i)_{i\geq0}$ be a sequence of positive numbers. Let
$\lambda>1$ and $p\geq1$, then 
\[
\biggl(\sum_{k\geq0}a_k\biggr)^p
\leq c(\lambda,p)\sum_{k\geq0} \lambda^ka_k^p,
\]
where $c(\lambda,p)=1$ for $p=1$ and
$c(\lambda,p)=\bigl(1-\lambda^{-1/(p-1)}\bigr)^{-(p-1)}$ otherwise.
\end{lem}

\begin{proof}
Simply apply H\"older inequality to
$\sum_{k\geq0}a_k=\int\lambda^{-k}a_kd\mu(k)$ where $\mu$ is the
discrete measure on the non-negative integers defined by
$\mu(k):=\lambda^k$.
\end{proof}

\begin{lem}
If $\xi_p>0$, then
$u\in L^p(Q_\emptyset)$.
\end{lem}

\begin{proof}
By Remark~\ref{rem:a5} and Lemma~\ref{lem:6.1}, for all $\lambda>1$,
\begin{equation*}
	\begin{split}
		\|u\|_{L^p}^p
		&\leq\int_{Q_\emptyset}\Bigl(\sum_{i\geq0}|u_{x_i}\psi_{x_i}(x)|\Bigr)^pdx
		\leq c(\lambda,p)\int_{Q_\emptyset}\sum_{i\geq0}\lambda^i|u_{x_i}\psi_{x_i}(x)|^pdx\\
		&=c(\lambda,p)\int_{Q_\emptyset}\sum_{j\in J}\lambda^{|j|}|u_j\psi_j(x)|^pdx
		\leq c_1(\lambda,p)\sum_{j\in J}\lambda^{|j|}|u_j|^p2^{(dp/2-d)|j|}\\
		&=c_1(\lambda,p)\sum_{i\geq0}\lambda^i2^{-\xi_pi}.\qedhere
	\end{split}
\end{equation*}
\end{proof}

We need to introduce an hypothesis on the function $u$, in that it
needs to show some sort of autosimilarity with respect to the wavelet
decomposition, as clarified below.

To do so, we need to introduce the sets of automorphisms on $J$, that is 
\[
	S\coloneqq \{\sigma:J\to
	J\,|\,\sigma(\emptyset)=\emptyset,\sigma(k)=\sigma(\bar\jmath)\text{ iff
	}k=\bar\jmath\}.
\]

\paragraph{Autosimilarity hypothesis.} 

For all $j\in J$ there exists $\sigma_j\in S$ such that for all $k\in
J$,
\[
u_{jk}\sim u_ju_{\sigma_j(k)}.
\]
Here with $\sim$ we intend that the absolute value of the ratio
between the two terms is uniformly bounded from above and below, away
from zero.

(Notice that the unique constant solution of an RCM trivially
satisfies this hypothesis.)

\begin{lem}\label{lem:6.3}
For all $n\geq0$, under autosimilarity hypothesis,
\[
\int_{Q_\emptyset}\Bigl|\sum_{|j|\geq n}u_j\psi_j(x)\Bigr|^pdx
\sim\|u\|_{L^p}^p2^{(\frac d2p-d)n}\sum_{|j|=n}|u_j|^p.
\]
\end{lem}

\begin{proof}
Any automorphism $\sigma\in S$ induces a measure-preserving map
$\sigma$ on $Q_\emptyset$, defined by $Q_\emptyset\ni
x=(x_0,x_1,x_2,\dots)\mapsto\sigma(x):=(\sigma(x_0),\sigma(x_1),\dots)$,
so that $\psi_{\sigma^{-1}(k)}(z)=\psi_k(\sigma(z))$.

Then, for any $j\in J$ with $|j|=n$, by the two hypothesis,
\begin{equation*}
\begin{split}
	\sum_{k\geq j}u_k\psi_k(x)
	&=\sum_{k\in J}u_{jk}\psi_{jk}(x)\\
	&\sim\sum_{k\in J}u_ju_{\sigma_j(k)}2^{\frac d2|j|}\psi_k(2^{|j|}x+\theta_j)
	\sim2^{\frac d2|j|}u_ju(\sigma_j(z)),
\end{split}
\end{equation*}
where $z=z(x,j)=2^{|j|}x+\theta_j$ spans $Q_\emptyset$ as $x$ spans $Q_j$. Thus
\begin{equation*}
	\begin{split}
		\int_{Q_\emptyset}\Bigl|\sum_{|j|\geq n}u_j\psi_j(x)\Bigr|^pdx
		&\sim2^{\frac d2np}\sum_{|j|=n}|u_j|^p\int_{Q_j}|u(\sigma_j(z))|^pdx\\
		&\sim\|u\|_{L^p}^p2^{(\frac d2p-d)n}\sum_{|j|=n}|u_j|^p.\qedhere
	\end{split}
\end{equation*}
\end{proof}

We decompose the difference appearing in $S_p$ as follows: 
\[
u(x)-u(y)
=\sum_{|j|<n}u_j(\psi_j(x)-\psi_j(y))
+\sum_{|j|\geq n}u_j\psi_j(x)-\sum_{|j|\geq n}u_j\psi_j(y).
\]
For the first terms, when $|j|<n$, 
\[
	|\psi_j(x)-\psi_j(y)|
	\approx|\nabla\psi_j||x-y|\mathbbm1_{Q_j}(x)
	\approx2^{(\frac d2+1)|j|}2^{-n}\mathbbm1_{Q_j}(x),
\]
and in particular
\[
	\int_{Q_\emptyset}\Bigl\langle\Bigl|\sum_{|j|=i}u_j(\psi_j(x)-\psi_j(y))\Bigr|^p\Bigr\rangle_ydx
	\approx2^{[(\frac d2+1)p-d]i-np}\sum_{|j|=i}|u_j|^p.
\]
Using Lemma~\ref{lem:6.3} to estimate the two remaining sums and putting
everything together, we get that for $\xi_p>0$,
\[
S_p(2^{-n})
\approx n^p2^{-np}\sum_{i=0}^n2^{[(\frac d2+1)p-d]i}\sum_{|j|=i}|u_j|^p
\approx 2^{-np}\sum_{i=0}^n2^{(p-\xi_p)i}
\approx 2^{-\min(p,\xi_p)n}
\]
hence we have the claimed result,
\[
-\lim_{n\to\infty}\frac1n\log_2S_p(2^{-n})=:
\zeta_p
=\min(p,\xi_p).
\]

\section*{Acknowledgements}
The authors were partially supported by Istituto Nazionale di Alta Matematica--Gruppo Nazionale Analisi Matematica, Probabilità e loro Applicazione, in the framework of the INdAM--GNAMPA Projects.

The authors would also like to thank the anonymous referees for their comments and corrections that substantially improved the paper.

\bibliographystyle{abbrv}
\bibliography{bibdyadictesi}

\end{document}